\documentclass[a4paper]{llncs}

\usepackage{pdfsync}
\usepackage{authcite}
\usepackage{proof}
\usepackage{prooftree}
\usepackage{stmaryrd}
\usepackage{amsmath}
\usepackage{amssymb}
\usepackage{latexsym}
\usepackage{xspace}
\usepackage{mymacros}
\usepackage[utf8]{inputenc}
\usepackage{pdfsync}
\usepackage{cancel}
\usepackage{color}
\usepackage{listings}
\usepackage{authcite}
\usepackage{enumitem}
\usepackage{yhmath}

\title{A Single-Assignment Translation for\\ Annotated Programs}

\author{Cláudio Belo Lourenço \and Maria João Frade \and\\ Jorge Sousa Pinto} 

\institute{HASLab/INESC TEC \& Universidade do Minho, Portugal}

\begin{document}
\maketitle 

% required for using listings with llncs
%\setcounter{chapter}{0}

\begin{abstract}
We present a translation of While programs annotated with loop invariants into
a dynamic single-assignment language with a dedicated iterating construct. We
prove that the translation is sound and complete. This is a companion report
to our paper \emph{Formalizing Single-assignment Program Verification: an
Adaptation-complete Approach}~\cite{PintoJS:sinapv}.
\end{abstract}

%%%%%%%%%%%%%%%%%%%%%%%%%%%%%%%%%%%%%%%%%%%%%%%%%%%%%%%%%%%%
\section{Introduction}

Deductive verification tools typically rely on the conversion of
code to a dynamic single-assignment (SA) form.  In~\cite{PintoJS:sinapv} we formalize this approach by proposing a sound and complete technique based on the translation of annotated programs to such an intermediate form,  and the generation of verification conditions from it. 
We introduce a notion of SA iterating program, as well as an  \emph{adaptation-complete} program logic (adequate for adapting specifications to local contexts) and an \emph{efficient verification conditions} generator (in the sense of Flanagan and Saxe). 

In this report we define a translation function that transforms an annotated program into SA from. 
The program produced conforms to the syntactic restrictions of an SA annotated program and preserves the operational semantics of the program in a sense that will be made precise later. Moreover, the translation function is lifted to Hoare triples, and  preserves derivability in a goal-directed system of Hoare logic, i.e, if a Hoare triple for an annotated program is derivable, then the translated triple is also derivable in that system (with derivations guided by the translated loop invariants).

The core result of the paper, which we prove in full detail, is that the defined translation conforms to the requirements, identified in~\cite{PintoJS:sinapv}, for the generation of verification conditions to be sound and complete for the initial program. More precisely, if a Hoare triple containing an annotated While program is translated into SA form by our translation, and verification conditions for this SA triple are then generated, then (i) if these verification conditions are successfully discharged, then the original triple is valid (the program is correct); and (ii) if the original program is correct, and moreover it is correctly annotated (i.e. its annotations allow for the Hoare triple to be derived) then the verification conditions are all valid. 

\medskip

The report is organised as follows. In the remaining of this section
we introduce some preliminary notation. In Section
\ref{sec:HoareLogic} we recall the necessary background material about
Hoare logic. In Section \ref{sec:translation} we define the
translation function $\tsa$, and illustrate its use by means of an
example in Section~\ref{sec:example}. In Section \ref{sec:soundness}
we prove that $\tsa$ is sound and preserves
$\Hg$-derivability. Section \ref{sec:conclusion} concludes the report.

%=====================================================
\subsection{Preliminary notation}\label{subsec:preliminaries}

First of all let us introduce the notation used for functions. 
Given a function $f$, $\dom(f)$ denotes the domain of $f$, $\rng(f)$
denotes the codomain of $f$.
% , and $\rest{f}{A}$ denotes the restriction of $f$ to the domain $A$, i.e.,
% $\dom(\rest{f}{A}) = \dom(f) \cap A$.
As usual, $f[x \mapsto a]$ denotes the function that maps $x$ to $a$ and any other value $y$
to $f(y)$. 

For finite functions we also use the following notation. $[]$
represents the empty function (whose domain is $\emptyset$), and
$[x_1\mapsto a_1,\ldots, x_n\mapsto a_n]$ represents the finite
function whose domain is $\{x_1,\ldots, x_n\}$ that maps each $x_i$ to
$a_i$. We also use the notation $[ x\mapsto g(x) \mid x\in A]$ to
represent the function with domain $A$ generated by $g$.

Given a total function  $f : X\to Y$ and a partial functions $g : X \pmap Y$, we define $f\oplus g : X \to Y$ as follows 
\[
(f\oplus g)(x) = 
\left\{  
  \begin{array}{ll} g(x) & \mbox{\rm ~ if~ $x\in\dom(g)$} \\
    f(x) & \mbox{\rm ~ if~ $x\not\in\dom(g)$} 
  \end{array}
\right.
\]
i.e., $g$ overrides $f$ but leaves $f$ unchanged at points where $g$
is not defined.

%%%%%%%%%%%%%%%%%%%%%%%%%%%%%%%%%%%%%%%%%%%%%%%%%%%%%%%%%%%%
\section{Hoare logic}\label{sec:HoareLogic}
We will work with Hoare logic for simple While programs. The logic deals
with the notion of correction w.r.t. a \emph{specification} that
consists of a \emph{precondition} -- an assertion that is assumed to
hold when the execution of the program starts -- and a
\emph{postcondition} -- an assertion that is required to hold when
execution stops.

%=====================================================
\subsection{Syntax}\label{subsec:syntax}

We consider a typical While language whose commands $C\in\com$ are
defined over a set of variables $x \in \V$ in the following way:
\begin{align*}
  \com \ni C\;  ::=  \; \skp \,\mid\,  \asgn{x}{e} \; \mid\;  C \sep C \;
               \mid\;  \ifte{b}{C}{C} \;  \mid\;  \while{b}{C} 
\end{align*}
We will not fix the language of program expressions $e \in {\bf Exp}$
and Boolean expressions $b \in {\bf Exp^{\bool}}$, both constructed
over variables from $\V$ (a standard instantiation is for ${\bf Exp}$
to be a language of integer expressions and $\bf Exp^{\bool}$
constructed from comparison operators over ${\bf Exp}$, together with
Boolean operators).
In addition to expressions and commands, we need formulas that express
properties of particular states of the program. Program assertions
$\phi, \theta, \psi \in\assert$ (preconditions and postconditions in
particular) are formulas of a first-order language obtained as an
expansion of $\bf Exp^{\bool}$.

We also require a class of formulas for specifying the behaviour of
programs. Specifications are pairs $(\phi, \psi)$, with
$\phi, \psi \in \assert$ intended as precondition and postcondition
for a program. The precondition is an assertion that is assumed to
hold when the program is executed, whereas the postcondition is
required to hold when its execution stops. A \emph{Hoare triple},
written as $\hoatri{\phi}{C}{\psi}$, expresses the fact that the
program $C$ conforms to the specification $(\phi, \psi)$.

%=====================================================
\subsection{Semantics}\label{subsec:semmantics}

We will consider an \emph{interpretation structure} $\M = (D, I)$ for
the vocabulary describing the concrete syntax of program
expressions. This structure provides an interpretation domain $D$ as
well as a concrete interpretation of constants and operators, given by
$I$.
The interpretation of expressions depends on a \emph{state}, which is
a function that maps each variable into its value. We will write
$\states = \V \to D$ for the set of states (note that this approach
extends to a multi-sorted setting by letting $\states$ become a
\emph{generic function space}). For $s\in\states$, $\overrd{s}{x}{a}$
will denote the state that maps $x$ to $a$ and any other variable $y$
to $s(y)$.
The interpretation of $e \in {\bf Exp}$ will be given by a function
$\bsem{e}_\M : \states \to D$, and the interpretation of
$b \in {\bf Exp^{\bool}}$ will be given by
$\bsem{b}_\M : \states \to \{\Lfalse, \Ltrue \}$. This reflects our
assumption that an expression has a value at every state (evaluation
always terminates without error) and that expression evaluation never
changes the state (the language is free of \emph{side effects}).
For the interpretation of assertions we take the usual interpretation
of first-order formulas, noting two facts: since assertions build on
the language of program expressions their interpretation also depends
on $\M$ (possibly extended to account for user-defined predicates and
functions), and states from $\Sigma$ can be used as \emph{variable
  assignments} in the interpretation of assertions. The interpretation
of the assertion $\phi \in \assert$ is then given by
$\bsem{\phi}_\M : \states \to \{\Lfalse, \Ltrue \}$, and we will write
$s \models \phi$ as a shorthand for $\bsem{\phi}_\M(s) = \Ltrue$. In
the rest of the paper we will omit the $\M$ subscripts for the sake of
readability; the interpretation structure will be left implicit.

\begin{figure}[t]
\begin{frameit}
  %\scriptsize
  \small
  \centering
  \begin{enumerate}
  \item $\eval{\skp}{s}{s}$
  \item $\eval{\asgn{x}{e}}{s}{\overrd{s}{x}{\bsem{e}(s)}}$
  \item if $\eval{C_1}{s}{s'}$ and $\eval{C_2}{s'}{s''}$, then
    $\eval{C_1 \sep C_2}{s}{s''}$
  \item if $\bsem{b}(s)= \Ltrue$ and $\eval{C_t}{s}{s'}$, then 
    $\eval{\ifte{b}{C_t}{C_f}}{s}{s'}$
  \item if $\bsem{b}(s)= \Lfalse$ and $\eval{C_f}{s}{s'}$,  then 
    $\eval{\ifte{b}{C_t}{C_f}}{s}{s'}$
  \item if $\bsem{b}(s)= \Ltrue$, $\eval{C}{s}{s'}$ and
    $\eval{\while{b}{C}}{s'}{s''}$, then 
    $\eval{\while{b}{C}}{s}{s''}$
  \item if $\bsem{b}(s)= \Lfalse$, then $\eval{\while{b}{C}}{s}{s}$
  \end{enumerate}
\end{frameit}
%\vspace{-1em}
  \caption{Evaluation semantics for While programs}
  \label{fig:eval-semantics}
\end{figure}

For commands, we consider a standard operational, natural style
semantics, based on a deterministic \emph{evaluation relation}
$\leadsto\, \subseteq \com \times \states \times \states$ (which again
depends on an implicit interpretation of program expressions). We will
write $\eval{C}{s}{s'}$ to denote the fact that if $C$ is executed in
the initial state $s$, then its execution terminates, and the final
state is $s'$. The usual inductive definition of this relation is
given in Figure~\ref{fig:eval-semantics}.

The intuitive meaning of the triple $\hoatri{\phi}{C}{\psi}$ is that
if the program $C$ is executed in an initial state in which the
precondition $\phi$ is true, then either execution of $C$ does not
terminate or if it does, the postcondition $\psi$ will be true in the
final state. Because termination is not guaranteed, this is
called a \emph{partial correctness} specification. Let us now define
formally the notion of validity for such a triple.

\begin{definition}\label{def:hoaretriple-validity}
  The Hoare triple $\hoatri{\phi}{C}{\psi}$ is said to be
  \emph{valid}, denoted $\models \hoatri{\phi}{C}{\psi}$, whenever for
  all $s,s'\in\states$, if $s \models \phi$ and $\eval{C}{s}{s'}$,
  then $s' \models \psi$.
\end{definition}

%=====================================================
\subsection{Hoare Calculus}\label{subsec:HoareCalculus}

Hoare~\cite{HoareCAR:axibcp} introduced an inference system for
reasoning about Hoare triples, which we will call system $\HL$ - see
Figure~\ref{fig:systems-H}. Note that the system contains
one rule (conseq) whose application is guarded by first-order
conditions. We will consider that reasoning in this system takes place
in the context of the \emph{complete theory} $\rm Th(\M)$ of the
implicit structure $\M$, so that when constructing derivations in
$\HL$ one simply checks, when applying the (conseq) rule, whether the
side conditions are elements of $\rm Th(\M)$. We will write
$\infHL \hoatri{\phi}{C}{\psi}$ to denote the fact that the triple is
derivable in this system with $\rm Th(\M)$.

\begin{figure}[t]
\begin{frameit}
  \centering
  \scriptsize
  % \begin{tabular}{c}
    \begin{minipage}[t]{.98\linewidth}
      $$
      \begin{array}{lc}
        % \mbox{\scriptsize \bf System $\HL$}
        % \\ \\ \

        (\mbox{skip})&
        \infer{\hoatri{\phi}{\skp}{\phi}}{}
        \\ \\ \\
        
        (\mbox{assign})&
        \infer{\hoatri{\psi[\subst{x}{e}]}{\asgn{x}{e}}{\psi}}{}
        \\ \\ \\

        (\mbox{seq})&
        \infer{\hoatri{\phi}{C_1\sep C_2}{\psi}}{
        \hoatri{\phi}{C_1}{\theta} 
        & \quad
          \hoatri{\theta}{C_2}{\psi} 
          }   
        \\ \\ \\
        
        (\mbox{if})&
        \infer{\hoatri{\phi}{\ifte{b}{C_t}{C_f}}{\psi}}{
        \hoatri{\phi \andd \boolemb{b}}{C_t}{\psi}
        & \quad
          \hoatri{\phi \andd \negg \boolemb{b}}{C_f}{\psi}
          } 
        \\ \\ \\

        (\mbox{while})&
        \infer{\hoatri{\theta}{\while{b}{C}}{\theta \andd \negg b}}{
        \hoatri{\theta \andd \boolemb{b}}{C}{\theta}
        }
        \\ \\ \\

        (\mbox{conseq})&
        \infer[\mbox{if } \ \begin{array}{ll} 
                                  \phi' \impl \phi\ \mbox{ and } \\ 
                                  \psi \impl \psi' 
                                \end{array}]
        {\hoatri{\phi'}{C}{\psi'}}{
        \hoatri{\phi}{C}{\psi}
        }
        \\ 
      \end{array}
      $$
    \end{minipage}
\end{frameit}
  % \end{tabular}
  \caption{\label{fig:systems-H} System $\HL$}
\end{figure}

System $\HL$
is sound w.r.t. the semantics of Hoare triples; it is also complete as
long as the assertion language is sufficiently expressive (a result
due to Cook~\cite{CookSA:soucaspv}).
One way to ensure this is to force the existence of a \emph{strongest
  postcondition} for every command and assertion. Let $C \in \com$ and
$\phi \in \assert$, and denote by ${\rm post}(\phi, C)$ the set of
states
$\{s' \in \states \mid \eval{C}{s}{s'} \mbox{ for some } s \in \states
\mbox{ such that } \bsem{\phi}(s) = \Ltrue \}$.
In what follows we will assume that the assertion language $\assert$
is \emph{expressive} with respect to the command language $\com$ and
interpretation structure $\M$, i.e., for every $\phi \in \assert$ and
$C \in \com$ there exists $\psi \in \assert$ such that
$s \models \psi$ iff $s \in {\rm post}(\phi, C)$ for any
$s \in \states$. The reader is directed to~\cite{AptKR:tenyhlsone} for
 details.

% \begin{proposition}\label{prop:soundnessH}
%    Let $C \in \com$ and $\phi, \psi \in \assert$.\
%    \begin{enumerate}
%    \item If $\infHL \hoatri{\phi}{C}{\psi}$, then
%      $\models \hoatri{\phi}{C}{\psi}$.
%    \item With $\assert$ expressive in the above sense, if
%      $\models \hoatri{\phi}{C}{\psi}$ then
%      $\infHL \hoatri{\phi}{C}{\psi}$.
%    \end{enumerate}
% \end{proposition}

\begin{proposition}[Soundness of system $\HL$]\label{prop:soundnessH}
  Let $C \in \com$ and $\phi, \psi \in \assert$. If $\infHL \hoatri{\phi}{C}{\psi}$, then
     $\models \hoatri{\phi}{C}{\psi}$.
\end{proposition}

\begin{proposition}[Completeness of system $\HL$]\label{prop:completenessH}
    Let $C \in \com$ and $\phi, \psi \in \assert$. With $\assert$
    expressive in the above sense, if 
     $\models \hoatri{\phi}{C}{\psi}$, then
     $\infHL \hoatri{\phi}{C}{\psi}$. 
\end{proposition}

Let $\FV{\phi}$ denote the set of free variables occurring in $\phi$.
The sets of \emph{variables occurring} and \emph{assigned} in the
program $C$ will be given by $\vars{C}$ and $\assd{C}$ according to
the next definition.

\begin{definition}\label{def:varsasgn}
Let $\vars{e}$ be the set of  variables occurring in expression $e$,
and $C\in\com$. 
\begin{itemize}
  \item The set of  \emph{variables occurring} in $C$, $\vars{C}$, is defined
as follows:
  { 
    \small
      \begin{align*}
        \vars{\skp}  & =  \emptyset \\
        \vars{\asgn{x}{e}} & =  \{x\} \cup \vars{e}\\
        \vars{C_1 \sep C_2} & =  \vars{C_1} \cup \vars{C_2} \\
        \vars{\ifte{b}{C_t}{C_f}}  & =  \vars{b} \cup \vars{C_t} \cup \vars{C_f} \\
        \vars{\while{b}{C}} & = \vars{b} \cup \vars{C}
      \end{align*}
  }
  \item The set of  \emph{variables assigned} in $C$, $\assd{C}$, is defined
as follows:
  {
    \small
    \begin{align*}
        \assd{\skp}  & =  \emptyset \\
        \assd{\asgn{x}{e}} & =  \{x\} \\
        \assd{C_1 \sep C_2} & =  \assd{C_1} \cup \assd{C_2} \\
        \assd{\ifte{b}{C_t}{C_f}} & = \assd{C_t} \cup \assd{C_f} \\
        \assd{\while{b}{C}} & =  \assd{C}
      \end{align*}
  }

  \item We will write
  $\SA{\phi}{C}$ to denote the fact that $C$ does not assign the
  variables occurring free in $\phi$, i.e.
  $\assd{C}\, \cap\, \FV{\phi} = \emptyset$.
\end{itemize}
\end{definition}

Triples in which the program do not assign variables from the precondition enjoy
the following property in $\HL$:

\begin{lemma}\label{lem:constancy}
  Let $\phi, \psi\in\assert$ and $C\in\com$, such that
  $\SA{\phi}{C}$. If $\infHL \hoatri{\phi}{C}{\psi}$, then
  $\infHL \hoatri{\phi}{C}{\phi \andd \psi}$.
% \begin{enumerate}
% \item $\infHL \hoatri{\phi}{C}{\phi}$.
% \item If $\infHL \hoatri{\phi}{C}{\psi}$, then $\infHL
%   \hoatri{\phi}{C}{\phi \andd \psi}$.
% \end{enumerate}
\end{lemma}
\begin{proof}
  By induction on the structure of $C$ one proves that
  $\infHL \hoatri{\phi}{C}{\phi}$.  Then it follows from
  $\infHL \hoatri{\phi}{C}{\psi}$, by the conjunction assertions rule
  (see for instance~\cite{ReynoldsJC:thepl}), that
  $\infHL \hoatri{\phi \land \phi}{C}{\phi \land \psi}$. We conclude
  by applying (conseq).  \hfill$\Box$
\end{proof}

\begin{figure}[t]
  \centering
  \footnotesize
  \begin{tabular}{|c|}
    \hline 
    \begin{minipage}[t]{1\linewidth}
      $$
      \begin{array}{ll}
         (\mbox{skip}) & \infer{\hoatri{\phi}{\skp}{\phi}}{}
        \\ \\ 
        
         (\mbox{assign}) \ \ \ & \infer{\hoatri{\psi[\subst{x}{e}]}{\asgn{x}{e}}{\psi}}{}
        \\ \\ 

         (\mbox{seq}) & \infer{\hoatri{\phi}{C_1\sep C_2}{\psi}}{
        \hoatri{\phi}{C_1}{\theta} 
        & \quad
          \hoatri{\theta}{C_2}{\psi} 
          }   
        \\ \\ 
        
         (\mbox{if}) & \infer{\hoatri{\phi}{\ifte{b}{C_t}{C_f}}{\psi}}{
        \hoatri{\phi \andd \boolemb{b}}{C_t}{\psi}
        & \quad
          \hoatri{\phi \andd \negg \boolemb{b}}{C_f}{\psi}
          } 
        \\ \\ 

        (\mbox{while}) &  \infer{\hoatri{\theta}{\while{b}{C}}{\theta \andd \negg b}}{
        \hoatri{\theta \andd \boolemb{b}}{C}{\theta}
        }
        \\ \\ 

         (\mbox{conseq}) & \infer[\mbox{if } \ \begin{array}{ll} 
                                  \phi' \impl \phi\ \mbox{ and } \\ 
                                  \psi \impl \psi' 
                                \end{array}]
        {\hoatri{\phi'}{C}{\psi'}}{
        \hoatri{\phi}{C}{\psi}
        }
        \\ 
      \end{array}
      $$
    \end{minipage}
    \\ \\
    \hline
  \end{tabular}
  \caption{\label{fig:systems-H} Systems $\HL$}
\end{figure}

%=====================================================
\subsection{Goal-directed logic}\label{subsec:goal-directedlogic}

We introduce a syntactic class $\Acom$ of \emph{annotated programs},
which differs from $\com$ only in the case of while commands, which
are of the form $\whileinv{b}{\theta}{C}$ where the assertion $\theta$
is an annotated loop invariant. Annotations do not affect the
operational semantics.  Note that for $C\in\Acom$, $\vars{C}$ includes
the free variables of the annotations in $C$. In what follows we will
use the auxiliary function $\Aerase{\cdot}$ that
erases all annotations from a program defined in the following definition.

\begin{definition}
  The function $\Aerase{\cdot} : \Acom \to \com$ is defined as follows:
  \begin{align*}
    \Aerase{\skp} & = \skp\\
    \Aerase{\asgn{x}{e}} & = \asgn{x}{e}\\
    \Aerase{C_1 \sep C_2} & = \Aerase{C_1}\sep \Aerase{C_2}\\
    \Aerase{\ifte{b}{C_t}{C_f}} & = \ifte{b}{\Aerase{C_t}}{\Aerase{C_f}}\\
    \Aerase{\whileinv{b}{\theta}{C}} & = \while{b}{\Aerase{C}}\\
  \end{align*}
\end{definition}

In Figure \ref{fig:systems-Hg} we present system $\Hg$, a
\emph{goal-directed} version of Hoare logic for triples containing
annotated programs. This system is intended for mechanical
construction of derivations: loop invariants are not invented but
taken from the annotations, and there is no ambiguity in the choice of
rule to apply, since a consequence rule is not present. The different
derivations of the same triple in $\Hg$ differ only in the
intermediate assertions used.

\begin{figure}[t]
  \centering
  \footnotesize
  \begin{tabular}{|c|}
    \hline 
    \begin{minipage}[t]{1\linewidth}
      $$
      \begin{array}{ll}
        (\mbox{skip}) &  \infer[\mbox{if } \phi \impl \psi]
        {\hoatri{\phi}{\skp}{\psi}}
        {}
        \\ \\   

        (\mbox{assign}) \ \ \ &  \infer[\mbox{if $\phi \impl \psi[\subst{x}{e}]$}]
        {\hoatri{\phi}{\asgn{x}{e}}{\psi}}
        {}
        \\ \\ 

        (\mbox{seq}) &  \infer{\hoatri{\phi}{C_1\sep C_2}{\psi}}{
        \hoatri{\phi}{C_1}{\theta} 
        & \quad
        \hoatri{\theta}{C_2}{\psi} 
        }   
        \\ \\ 
        
         (\mbox{if}) & \infer{\hoatri{\phi}{\ifte{b}{C_t}{C_f}}{\psi}}{
        \hoatri{\phi \andd \boolemb{b}}{C_t}{\psi}
        & \quad
        \hoatri{\phi \andd \negg \boolemb{b}}{C_f}{\psi}
        } 
        \\ \\ 

        (\mbox{while}) &  \infer[\mbox{if } \ \begin{array}{ll} 
                                  \phi \impl \theta\ \mbox{ and } \\ 
                                  \theta \andd \negg b \impl \psi 
                                \end{array}]
        {\hoatri{\phi}{\whileinv{b}{\theta}{C}}{\psi}}{
        \hoatri{\theta \andd \boolemb{b}}{C}{\theta}
        }
        \\ 
      \end{array}
      $$
    \end{minipage}
    \\ \\
    \hline
  \end{tabular}
  \caption{\label{fig:systems-Hg} System $\Hg$ }
\end{figure}

The following can be proved by induction on the
derivation of $\infHg \hoatri{\phi}{C}{\psi}$.

\begin{proposition}[Soundness of $\Hg$]\label{prop:Hg-sound-loops}
  Let $C \in \Acom$ and $\phi, \psi \in \assert$. 
  $
  \mbox{ If }  \infHg \hoatri{\phi}{C}{\psi}  \mbox{ then
  }  \infHL \hoatri{\phi}{\Aerase{C}}{\psi}
  $.
\end{proposition}
The converse implication does not hold, since the annotated invariants
may be inadequate for deriving the triple. Instead we need the
following definition:
\begin{definition}\label{def:correctly-annot}
  Let $C \in \Acom$ and $\phi, \psi \in \assert$. We say that $C$ is 
  \emph{correctly-annotated} w.r.t. $(\phi, \psi)$ if $\infHL
  \hoatri{\phi}{\Aerase{C}}{\psi}$ implies $\infHg
  \hoatri{\phi}{C}{\psi}$.
\end{definition}

The following lemma states the admissibility of the consequence rule
in $\Hg$.

\begin{lemma}\label{lemma:Hg-conseq}
  Let $C \in \Acom$ and $\phi, \psi, \phi', \psi' \in \assert$ such
  that $\infHg \hoatri{\phi}{C}{\psi}$, $\models \phi' \impl \phi$,
  and $\models \psi \impl \psi'$. Then $\infHg
  \hoatri{\phi'}{C}{\psi'}$.
\end{lemma}

It is possible to write an algorithm,
known as a \emph{verification conditions generator} (VCGen), that
simply collects the side conditions of a derivation without actually
constructing it. $\Hg$ is agnostic with respect to a strategy for
propagating assertions, but the VCGen necessarily imposes one such
strategy~\cite{GordonMJC:forh}.

%=====================================================
\subsection{Single-assignment programs}\label{subsec:saprograms}

Translation of code into Single-Assignment (SA) form has been part of
the standard compilation pipeline for decades now; in such a program
each variable is assigned at most once. The fragment
$\asgn{x}{10}\sep \asgn{x}{x+10}$ could be translated as
$\asgn{x_1}{10}\sep \asgn{x_2}{x_1+10}$, using a different ``version
of $x$'' variable for each assignment. In this paper we will use a
dynamic notion of single-assignment (DSA) program, in which each
variable may occur syntactically as the left-hand side of more than
one assignment instruction, as long as it is not assigned more than
once \emph{in each execution}. For instance the fragment
$\ifte{x>0}{\asgn{x}{x+10}}{\skp}$ could be translated into DSA form
as $\ifte{x_0>0}{\asgn{x_1}{x_0+10}}{\asgn{x_1}{x_0}}$. Note that the
\emph{else} branch cannot be simply $\skp$, since it is necessary to
have a single version variable (in this case $x_1$) representing $x$
when exiting the conditional.

In the context of the guarded commands language, it has been shown
that verification conditions for \emph{passive programs} (essentially
DSA programs without loops) can be generated avoiding the exponential
explosion problem.
However, a proper single-assignment imperative language in which
programs with loops can be expressed for the purpose of verification
does not exist. In what follows we will introduce precisely such a
language based on DSA form.

\begin{definition} The set $\rnm \subseteq \com$ of \emph{renamings} consists
  of all programs of the form
  $\{\asgn{x_1}{y_1}\sep \ldots \sep \asgn{x_n}{y_n}\}$ such that all
  $x_i$ and $y_i$ are distinct.
  \label{def:rnm}
\end{definition}

A renaming $\R = \{\asgn{x_1}{y_1}\sep \ldots \sep \asgn{x_n}{y_n}\}$
represents the finite bijection
$[x_1 \mapsto y_1,\ldots,x_n \mapsto y_n]$, which we will also denote
by $\R$. 
% We will write $\dom(\R)$ and $\rng(\R)$ to denote the domain
% and range of $\R$, respectively.
%
% Moreover $\R^{-1}$ denotes both the inverse finite function
% of $\R$ and the sequence of assignments
% $\asgn{y_1}{x_1} \sep\ldots\sep \asgn{y_n}{x_n}$. Note that $\R^{-1}$
% is well defined since $\R$ is injective.
  %
Furthermore, %if $\phi \in \assert$,
$\R(\phi)$ will denote the assertion that results from applying the
substitution $[\subst{x_1}{y_1},\ldots,\subst{x_n}{y_n}]$ to $\phi$.
Also, for $s\in\states$ we define the state $\R(s)$ as follows: $\R(s)(x) =
s(\R(x))$ if $x\in\dom(\R)$, and $\R(s)(x) = s(x)$ otherwise.
% Note that, as it can be easily proved by inspection on the evaluation relation,
% one has $\eval{\R}{s}{\R(s)}$.

%%%%%%%%%%%%%%%%%%
% A renaming $\R = \{\asgn{x_1}{y_1}\sep \ldots \sep \asgn{x_n}{y_n}\}$
% represents the finite bijection
% $[x_1 \mapsto y_1,\ldots,x_n \mapsto y_n]$, which we will also denote
% by $\R$. 
% Moreover $\R^{-1}$ denotes both the inverse finite function
%  of $\R$ and the sequence of assignments
%  $\asgn{y_1}{x_1} \sep\ldots\sep \asgn{y_n}{x_n}$. Note that $\R^{-1}$
%  is well defined since $\R$ is injective.
%   %
% Furthermore, %if $\phi \in \assert$, 
% we will write $\R(\phi)$ to denote
% the assertion that results from applying the substitution
% $[\subst{x_1}{y_1},\ldots,\subst{x_n}{y_n}]$ to $\phi$. 
% %
% Also, for $s\in\states$, we define the state $\R(s)$ as follows: $\R(s)(x) =
% s(\R(x))$ if $x\in\dom(\R)$, and $\R(s)(x) = s(x)$ otherwise.

\begin{lemma}\label{lemma:rnm}
Let $\R \in\rnm$, $\phi, \psi\in\assert$ and $s\in\states$.
  \begin{enumerate}
  \item $\eval{\R}{s}{\R(s)}$
  \item $\bsem{\R(\phi)}(s) = \bsem{\phi}(\R(s))$
  \item $\models\hoatri{\phi}{\R}{\psi}$ ~ iff ~  $\models \phi \impl \R(\psi)$.
  \end{enumerate}
\end{lemma}
\begin{proof}
  1. By inspection on the evaluation relation. 2. By structural
  induction on the interpretation assertions. 3. Follows directly from
  1 and 2.
\hfill$\Box$
\end{proof}

\begin{definition}\label{def:SA-for-program}
Let $\AFORSAcom$ be the class of \emph{annotated single-assignment
  programs}. Its abstract syntax is defined by 
$$
C\ ::=  \ \skp \mid C \sep C \mid \asgn{x}{e} \mid  \ifte{b}{C}{C} \mid \forinv{\Ic}{b}{\Uc}{\theta}{C}
$$
\vspace{-8mm}
  \begin{itemize}\small
  \item[] \hspace{-7mm} with the following restrictions:
  \item $\skp \in \AFORSAcom$
  \item $\asgn{x}{e} \in \AFORSAcom$ if  $x\not\in\vars{e}$
  \item $C_1 \sep C_2 \in \AFORSAcom$ if $C_1, C_2 \in \AFORSAcom$ and
    $\vars{C_1}\cap \assd{C_2} = \emptyset$
  \item $\ifte{b}{C_t}{C_f} \in \AFORSAcom$ if $C_t, C_f \in
    \AFORSAcom$ and $\vars{b}\cap (\assd{C_t} \cup \assd{C_f}) = \emptyset$
  \item $\forinv{\Ic}{b}{\Uc}{\theta}{C} \in \AFORSAcom$ if $C \in
    \AFORSAcom$, $\Ic, \Uc \in \rnm$, $\assd{\Ic}=\assd{\Uc}$,
    $\rng(\Uc) \subseteq \assd{C}$, and $(\vars{\Ic}\cup \vars{b}\cup \FV{\theta})\cap\assd{C}
    =\emptyset$
  \end{itemize}
\noindent where the sets of variables occurring and assigned in the
program are extended with the case of for-commands as follows:
$$
{\small
  \begin{array}{lcl}
\vars{\forinv{\Ic}{b}{\Uc}{\theta}{C}}\;  &=& \vars{\Ic} \cup \vars{b}
  \cup \FV{\theta} \cup \vars{C}
  \\ 
 \assd{\forinv{\Ic}{b}{\Uc}{\theta}{C}}  &=& \assd{\Ic} \cup \assd{C}
  \end{array}
}
$$
\end{definition}  
Note that the definition of $\SA{\phi}{C}$ is naturally lifted to
annotated programs.

The definition is straightforward except in the case of loops. In a
strict sense it is not possible to write iterating programs in DSA
form. So what we propose here is a syntactically controlled violation
of the single-assignment constraints that allows for structured
reasoning.  Loop bodies are SA blocks, but loops contain a renaming
$\Uc$ (to be executed after the body) that is free from
single-assignment restrictions. The idea is that the ``initial
version'' variables of the loop body (the ones used in the loop
condition) are updated by the $\Uc$ renaming, which transports values
from one iteration to the next, and is allowed to assign variables
already assigned in $\Ic$ or occurring in $C$ or $b$. The
initialization code $\Ic$ on the other hand contains a renaming
assignment that runs exactly once, even if no iterations take
place. This ensures that the initial version variables always contain
the output values of the loop: $\Uc$ is always executed after every
iteration, and in the case of zero iterations they
have been initialized by $\Ic$.

Consider the factorial program shown below on the left.  The counter
$i$ ranges from $1$ to $n$ and the accumulator $f$ contains at each
step the factorial of $i-1$. The program is annotated with an
appropriate loop invariant; it is easy to show that the program is
correct w.r.t. the specification
$(n \geq 0 \andd n=n_{\it aux}, f = n_{\it aux}!)$.

\vspace{-1em}
\begin{tabular}{ll}
\begin{minipage}[t]{.5\textwidth} 
\footnotesize
$$
\begin{array}{l}
  % \{n \geq 0 \andd n=n_{\it aux}\} \\
    \asgn{f}{1} \sep \\
     \asgn{i}{1} \sep \\
     \whileinv{i \leq n}{f = (i-1)! \andd i \leq n+1}{\\
       \block{\\
         \quad \asgn{f}{f * i} \sep \\
         \quad \asgn{i}{i + 1}  \\
       }} 
  % \{f = n_{\it aux}!\} \\
\end{array}
$$
\end{minipage}
&\qquad
\begin{minipage}[t]{.3\textwidth} 
\footnotesize
$$
\begin{array}{l}
      \asgn{f_1}{1} \sep \asgn{i_1}{1} \sep \\
      \Ic\\
     \while{(i_{a0} \leq n)}{\\
       \block{\\
         \quad \asgn{f_{a1}}{f_{a0} * i_{a0}}\sep \asgn{i_{a1}}{i_{a0} + 1}\sep  \\
         \quad \Uc\\
       }}
\end{array}
$$
\end{minipage}
\end{tabular}\\
On the right we show the same code with the blocks converted to SA
form. The variables in the loop are indexed with an 
`$a$', and then sequentially indexed with integer numbers as assignments
take place (any  fresh variable names would do). The initial
version variables of the loop body $f_{a0}$ and $i_{a0}$ are the ones used
in the Boolean expression, which is evaluated at the beginning of each
iteration. We have placed in the code the required renamings $\Ic$ and
$\Uc$, and it should be easy to instantiate them.  $\Ic$ should be
defined as $\asgn{i_{a0}}{i_1} \sep \asgn{f_{a0}}{f_1}$, and $\Uc$ as
$\asgn{i_{a0}}{i_{a1}}\sep \asgn{f_{a0}}{f_{a1}}$. Note that without $\Uc$ the new
values of the counter and of the accumulator would not be transported
to the next iteration. The initial version variables $i_{a0}$ and $f_{a0}$
are the ones to be used after the loop to access the value of the
counter and accumulator. A specification for this
program could be written  as
$(n \geq 0 \andd n=n_{\it aux}, f_{a0} = n_{\it aux}!)$.

It is straightforward to convert this pseudo-SA code to a program that
is in accordance with Definition~\ref{def:SA-for-program}, with a
\emph{for} command encapsulating the structure of the SA loop. The
required invariant annotation uses the initial version variables:
$$
{\scriptsize
\begin{array}{l}
% \{n \geq 0 \andd n=n_{\it aux}\} \\
    \asgn{f_1}{1} \sep 
         \asgn{i_1}{1} \sep \\
    \forinv{\block{\asgn{i_{a0}}{i_1} \sep \asgn{f_{a0}}{f_1}}}{i_{a0} \leq
    n}{\block{\asgn{i_{a0}}{i_{a1}} \sep \asgn{f_{a0}}{f_{a1}}}}{f_{a0} = (i_{a0}-1)! \andd i_{a0} \leq n+1}{\\
      \block{\\
        \quad \asgn{f_{a1}}{f_{a0} * i_{a0}} \sep 
                   \asgn{i_{a1}}{i_{a0} + 1}  \\
      }}\\
% \{f_{a0} = n_{\it aux}!\} \\
\end{array}
}
$$

\begin{definition}
 The function $\Tinv : \AFORSAcom \to \Acom$ translates SA programs to
(annotated) While programs as follows:
$$
\begin{array}{rcl}
  \Tinv(\skp) & = & \skp \\
  \Tinv(\asgn{x}{e}) & = & \asgn{x}{e} \\
  % \Tinv(\assume{\theta}) & = & \assume{\theta}\\
  % \Tinv(\assrt{\theta}) & = & \assrt{\theta}\\
  \Tinv(C_1 \sep C_2) & = & \Tinv(C_1) \sep \Tinv(C_2) \\
  \Tinv(\ifte{b}{C_t}{C_f}) & = & \ifte{b}{\Tinv(C_t)}{\Tinv(C_f)} \\
  \Tinv(\forinv{\Ic}{b}{\Uc}{\theta}{C}) & = & \Ic \sep \whileinv{b}{\theta}{\{\Tinv(C) \sep \Uc \}} \\
\end{array}
$$ 
\end{definition}

A translation of annotated programs into SA form (as ilustrated by the
factorial example) must of course
abide by the syntactic restrictions of $\AFORSAcom$, with additional
requirements of a semantic nature. In particular, the translation will
annotate the SA program with loop invariants (produced from those
contained in the original program), and $\Hg$-derivability guided by
these annotations must be preserved. On the other hand, the
translation must be sound: it will not translate invalid triples into
valid triples. Both these notions are expressed by translating back to
While programs.

\begin{definition}[SA translation]\label{def:Htriples-transl-loops}\
  A function
  $\T : \assert \times \Acom \times \assert \to \assert \times
  \AFORSAcom \times \assert$
  is said to be a \emph{single-assignment translation} if when
  $\T(\phi, C, \psi) = (\phi', C', \psi')$ we have $\SA{\phi'}{C'}$,
  and the following both hold:
  \begin{enumerate}
  \item If  $\: \models \hoatri{\phi'}{\Aerase{\Tinv(C')}}{\psi'}$, 
    then  $\: \models \hoatri{\phi}{\Aerase{C}}{\psi}$.
  \item If  $\: \infHg \hoatri{\phi}{C}{\psi}$, then
 $\: \infHg \hoatri{\phi'}{\Tinv(C')}{\psi'}$.
  \end{enumerate}
\end{definition}

The remmaning of this report is devoted to the definition of a
concrete SA translation function and to proving that the function
defined conforms Definition~\ref{def:Htriples-transl-loops}.

%%%%%%%%%%%%%%%%%%%%%%%%%%%%%%%%%%%%%%%%%%%%%%%%%%%%%%%%%%%%
\section{A translation to single-assigment form}\label{sec:translation}

In this section we define a translation function that transforms an
annotated program into SA form.  We start by introducing some
auxiliary definitions to deal with variable versions.
Without loss of generality, we will assume that the universe of variables of
the SA programs consists of two parts: the \emph{variable identifier} and a
\emph{version}. A version is a non-empty list of positive numbers.  We let
$\Vsa = \V \times \NAT^+$  be the set of SA variables, and we will write $x_l$
to denote $(x,l)\in\Vsa$. We write $\SAstates = \Vsa \to D$ for the set of
states, with $D$ being the interpretation domain.

Consider the \emph{version function}  $\Vcal : \V \to \NAT^+$. The function
$\tosa{\Vcal} : \V\to\Vsa$ is such that  $\tosa{\Vcal}(x)= x_{\Vcal(x)}$.
$\tosa{\Vcal}$ is lifted to ${\bf Exp}$ and $\assert$ in the obvious way,
renaming the variables according to $\Vcal$.
Let $s\in\states$ and $\Vcal : \V \to \NAT^+$. We define $\Vcal(s) \in
\Vsa\pmap D$ as the partial function $[  \tosa{\Vcal}(x) \mapsto s(x) ~|~ x
\in \V ]$.
% {\color{blue}
% Moreover, for $s'\in\SAstates$, $s' \oplus \Vcal(s)$ denotes the
% overriding of $s'$ by $\Vcal(s)$.
% }

%%%%%%%%%%
% In this section we define a translation function
% that transforms an annotated program into SA from.  
% Before going into the definition, we have to introduce
% some auxiliary definitions around the concepts of variable versions.

% Without loss of generality, we will assume that the universe of variables of the SA programs consists of two parts: the \emph{variable identifier} and a \emph{version}. A version is a non-empty list of positive numbers. 
% We let $\Vsa = \V \times \NAT^+$  be the set of SA variables, and we will write $x_l$ to denote $(x,l)\in\Vsa$.
% We write $\SAstates = \Vsa \to D$ for the set of states, with $D$ being the interpretation domain.

% Consider the \emph{version function}  $\Vcal : \V \to \NAT^+$. The function $\tosa{\Vcal} : \V\to\Vsa$
% is such that  $\tosa{\Vcal}(x)= x_{\Vcal(x)}$. $\tosa{\Vcal}$ is lifted to ${\bf Exp}$ and $\assert$ in the obvious way, renaming the variables according to $\Vcal$.

% Let $s\in\states$ and $\Vcal : \V \to \NAT^+$. We define $\Vcal(s) \in \Vsa\pmap D$ as the partial function $[  \tosa{\Vcal}(x) \mapsto s(x) ~|~ x \in \V ]$. 

\begin{figure}[]
%  \tiny
\scriptsize
\begin{frameit}
  \begin{align*}
    \tsa&: (\Var \to \mathbb{N}^+) \times \Acom \to (\Var \to \mathbb{N}^+) \times \AFORSAcom\\
    \tsa(\Vcal,\skp) & = (\Vcal,\skp) \\
    \tsa(\Vcal, \asgn{x}{e}) & = (\Vcal[x \mapsto \nextt(\Vcal(x))], 
                      \asgn{x_{\nextt(\Vcal(x))}}{\tosa{\Vcal}(e))} \\
    \tsa(\Vcal, C_1;C_2) & = (\Vcal'',C_{1}'; C_{2}') \\
      & \arraycolsep=0.5pt\def\arraystretch{1.5}
        \begin{array}{lll}
          \mbox{ where } & (\Vcal',C_{1}') &= \tsa(\Vcal, C_1) \\
                         & (\Vcal'',C_{2}') &= \tsa(\Vcal', C_2) 
        \end{array}\\
    \tsa(\Vcal, \ifte{b}{C_t}{C_f}) & = (\supp(\Vcal',\Vcal''),
                        \ifte{\tosa{\Vcal}(b)}{C_t'; \mergee(\Vcal',\Vcal'')}{C_f'; \mergee(\Vcal'',\Vcal')})   \\
      & \arraycolsep=0.5pt\def\arraystretch{1.5}
        \begin{array}{lll}
          \mbox{ where } & (\Vcal',C_t') &= \tsa(\Vcal,C_t) \\
                         & (\Vcal'',C_f') &= \tsa(\Vcal,C_f)
        \end{array}\\
    \tsa(\Vcal, \whileinv{b}{\theta}{C}) 
        & = (\Vcal''', \forinv{\Ic}{\tosa{\Vcal'}(b)}{\Uc}{\tosa{\Vcal'}(\theta)}{\{C'\}}; \upd(\dom(\Uc))) \\
        & \arraycolsep=0.5pt\def\arraystretch{1.5}
          \begin{array}{lll}
            \mbox{ where } & \Ic &= [\asgn{x_{\new(\Vcal(x))}}{x_{\Vcal(x)}} \mid x \in \assd{C}] \\
                 & \Vcal' &= \Vcal[x \mapsto \new(\Vcal(x)) \mid x \in \assd{C}] \\ 
                 & (\Vcal'',C') &= \tsa(\Vcal',C) \\
                 & \Uc &= [\asgn{x_{\new(\Vcal(x))}}{x_{\Vcal''(x)}} \mid x \in \assd{C}] \\
                 & \Vcal''' &= \Vcal''[x \mapsto \jump(l) \mid x_l \in \dom(\Uc)]
          \end{array}
  \end{align*}
  % \begin{align*}
  %   \tsa&: (\Var \to \mathbb{N}^+) \times \Acom \to (\Var \to \mathbb{N}^+) \times \AFORSAcom\\
  %   \tsa(\Vcal,\skp) & = (\Vcal,\skp) \\
  %   \tsa(\Vcal, \asgn{x}{e}) & = (\Vcal[x \mapsto \nextt(\Vcal(x))], 
  %                     \asgn{x_{\nextt(\Vcal(x))}}{\tosa{\Vcal}(e))} \\
  %   \tsa(\Vcal, C_1;C_2) & = (\Vcal'',C_{1}'; C_{2}') \\
  %     where\ & (\Vcal',C_{1}') = \tsa(\Vcal, C_1) \\
  %            & (\Vcal'',C_{2}') = \tsa(\Vcal', C_2) \\
  %   \tsa(\Vcal, \ifte{b}{C_t}{C_f}) & = (\supp(\Vcal',\Vcal''),
  %                       \ifte{\tosa{\Vcal}(b)}{C_t'; \mergee(\Vcal',\Vcal'')}{C_f'; \mergee(\Vcal'',\Vcal')})   \\
  %     where\ & (\Vcal',C_t') = \tsa(\Vcal,C_t) \\
  %            & (\Vcal'',C_f') = \tsa(\Vcal,C_f) \\
  %   \tsa(\Vcal, \whileinv{b}{\theta}{C}) 
  %       & = (\Vcal''', \forinv{\Ic}{\tosa{\Vcal'}(b)}{\Uc}{\tosa{\Vcal'}(\theta)}{\{C'\}}; \upd(\dom(\Uc))) \\
  %     where\ & \Ic = [\asgn{x_{\new(\Vcal(x))}}{x_{\Vcal(x)}} \mid x \in \assd{C}] \\
  %            & \Vcal' = \Vcal[x \mapsto \new(\Vcal(x)) \mid x \in \assd{C}] \\ 
  %            & (\Vcal'',C') = \tsa(\Vcal',C) \\
  %            & \Uc = [\asgn{x_{\new(\Vcal(x))}}{x_{\Vcal''(x)}} \mid x \in \assd{C}] \\
  %            & \Vcal''' = \Vcal''[x \mapsto \jump(l) \mid x_l \in \dom(\Uc)]
  % \end{align*}
\hrulefill \\
  \begin{minipage}[t]{.40\linewidth}
    \begin{align*}
      &\nextt : \mathbb{N}^+ \to \mathbb{N}^+\\
      &\nextt\ (h:t) = (h + 1):t\\
      \\
      &\new : \mathbb{N}^+ \to \mathbb{N}^+\\
      &\new\ l = 1:l\\
      \\
      &\jump : \mathbb{N}^+ \to \mathbb{N}^+ \\
      &\jump\ (i:j:t) = (j+1):t \\
      \\
      &  (h:t) \prec (h':t') = h < h' \\
    \end{align*}
  \end{minipage}
  \begin{minipage}[t]{.45\linewidth}
    \begin{align*}
      &\supp : (\Var \to \mathbb{N}^+)^2 \to (\Var \to \mathbb{N}^+)\\
      &\supp\ (\Vcal,\Vcal')(x) =  \left\{ 
        \begin{array}{l l}
          \Vcal(x) & \caseif\ \Vcal'(x) \prec \Vcal(x) \\
          \Vcal'(x) & \otherwise
        \end{array} \right.\\
      \\
      &\mergee : (\Var \to \mathbb{N}^+)^2 \to \rnm \\
      &\mergee\ (\Vcal,\Vcal') = [\asgn{x_{\Vcal'(x)}}{x_{\Vcal(x)}} \mid
                                          x \in \V \wedge \Vcal(x) \prec \Vcal'(x)]\\
      \\
      &\upd : \mathcal{P}(\Varsa) \to \rnm \\
      &\upd\ (X) = [\asgn{x_{\jump(l)}}{x_l} \mid x_l \in X ] \\
    \end{align*}
  \end{minipage}

\hrulefill \\
  \begin{align*}
 &   \tsa : \assert \times \Acom \times \assert \to \assertsa \times
   \AFORSAcom \times \assertsa \\ 
  &  \tsa(\phi,C,\psi)  = (\tosa{\Vcal}(\phi),C',\tosa{\Vcal'}(\psi)) \\
   & \hspace{1.5cm}  \mbox{where} \ (\Vcal',C') = \tsa(\Vcal,C)  \mbox{ ,  for
     some } \Vcal \in \Var \to \mathbb{N}^+
  \end{align*}

\end{frameit}
 \caption{\label{fig:sa-translation} SA translation function}
\end{figure}

The translation function $\tsa$ is presented in
Figure~\ref{fig:sa-translation} (top).  The function $\tsa$ receives
the initial version of the variable identifier and the annotated
program, and returns a pair with the final version of each variable
identifier and the SA translated program.
%Note that the initial versions of the variables are not fixed - any version will do.
The definition of $\tsa$ relies on various auxiliary functions that
deal with the version list and version functions, and also generate
renaming commands.  The functions are defined using Haskell-like
syntax. We give a brief description of each one: 
\begin{itemize}

\item $\nextt$ increments the first element of a variable version. 

  % It is used to create a different variable name when a variable is
  % assigned.

\item $\new$ appends a new element at the head of a variable
  version. 

  % Useful when entering a $\whilekey$ loop and to create a different
  % context.

\item $\jump$ is used to merge the first two elements of a variable
  version. It is used when exiting a loop, in order to return to the
  previous context.

\item $\supp$ receives two functions $\Vcal$ and $\Vcal'$, and returns
  a new one that returns the highest version for each variable
  (depending on whether it is in $\Vcal$ or in $\Vcal'$). A variable
  version $x:xs$ is higher than $y:ys$ if $x > y$.

\item $\mergee$ receives two functions $\Vcal$ and $\Vcal'$ and
  returns a $\rnm$ (see Definition~\ref{def:rnm}) containing
  assignments of the form $\asgn{\tosa{\Vcal'}(x)}{\tosa{\Vcal}(x)}$,
  for each $x \in \dom(\Vcal)$ such that $\Vcal(x) \prec \Vcal'(x)$.

\item $\upd$ fetches the appropriate version of a variable to be used
  after a loop.
\end{itemize}

For the sake of simplicity, we assume that the renaming sequences
$\Ic$ and $\Uc$, defined in the case of while commands, follow some
predefined order established over $\V$ (any order will do).

\section{Example}
\label{sec:example}

We will now illustrate the use of the SA translation by showing the
result of applying it to the following Hoare triple containing the
factorial program (let us call it $\mathsf{FACT}$). To illustrate how nested loops are handled, only
sum arithmetic instructions are used in the program, and
multiplication is implemented by a loop.

$$
\begin{array}{l}
\hoatrinl{n \geq 0 \andd aux=n}{
  \asgn{f}{1} \sep \\
  \asgn{i}{1} \sep \\
  \whileinv{i \leq n}{f = (i-1)! \andd i \leq n+1}{\\
    \block{\\
      \quad \asgn{j}{1} \sep \\
      \quad \asgn{r}{0} \sep \\
      \quad \whileinv{j \leq i}{j \leq i + 1 \wedge r = f * (j-1)}{\\
        \quad \block{\\
          \qquad \asgn{r}{r + f} \sep \\
          \qquad \asgn{j}{j + 1}  \\
          \quad }} \sep \\
      \quad \asgn{f}{r} \sep \\
      \quad \asgn{i}{i + 1}  \\
    }}}
{f = aux!}
\end{array}
$$

Below we show the result of applying the program-level $\tsa$ function
to the above program, taking as initial version function
$\Vcal$ that maps every variable to the list containing the sole
element 0 (note that any version function could be used).  For the sake 
of presentation, index lists will be depicted
using `.'  as a separator and omitting the empty list constructor. The
translated function, which we will call $\mathsf{FACT^{sa}}$ is as follows:

$$
\begin{array}{l}
  \asgn{f_1}{1} \sep \\
  \asgn{i_1}{1} \sep \\
  \forinvT{\block{\asgn{j_{1.0}}{j_0} \sep \asgn{r_{1.0}}{r_0} \sep \asgn{f_{1.1}}{f_1} \sep \asgn{i_{1.1}}{i_1}}}
         {i_{1.1} \leq n_0}
         {\block{\asgn{j_{1.0}}{j_{3.0}} \sep \asgn{r_{1.0}}{r_{3.0}} \sep \asgn{f_{1.1}}{f_{2.1}} \sep \asgn{i_{1.1}}{i_{2.1}}}}
         {f_{1.1} = (i_{1.1}-1)! \andd i_{1.1} \leq n_0+1}{
    \block{\\
      \quad \asgn{j_{2.0}}{1} \sep \\
      \quad \asgn{r_{2.0}}{0} \sep \\
      \quad \forinvTT{\block{\asgn{r_{1.2.0}}{r_{2.0}} \sep \asgn{j_{1.2.0}}{j_{2.0}} }}
                   {j_{1.2.0} \leq i_{1.1}}
                   { \block{\asgn{r_{1.2.0}}{r_{2.2.0}} \sep \asgn{j_{1.2.0}}{j_{2.2.0}}}}
                   { j_{1.2.0} \leq i_{1.1} + 1 \wedge r_{1.2.0} = f_{1.1} * (j_{1.2.0}-1)}{\\
                     \quad \block{\\
                       \qquad \asgn{r_{2.2.0}}{r_{1.2.0} + f_{1.1}} \sep \\
                       \qquad \asgn{j_{2.2.0}}{j_{1.2.0} + 1}  \\
                       \quad }} \sep \\
      \quad \asgn{r_{3.0}}{r_{1.2.0}} \sep \\ 
      \quad \asgn{j_{3.0}}{j_{1.2.0}} \sep \\
      \quad \asgn{f_{2.1}}{r_{3.0}} \sep \\
      \quad \asgn{i_{2.1}}{i_{1.1} + 1}  \\
    }}\sep\\
    \asgn{j_1}{j_{1.0}} \sep \\
    \asgn{r_1}{r_{1.0}} \sep \\
    \asgn{f_2}{f_{1.1}} \sep \\ 
    \asgn{i_2}{i_{1.1}}
\end{array}
$$

In addition to the SA program, $\tsa$ returns a version function
$\Vcal'$ which is $\Vcal[f \mapsto 2, i \mapsto 2, j \mapsto 1, r \mapsto 1]$.  
The initial and final versions functions $\Vcal$ and
$\Vcal'$ will be applied to the precondition and postcondition to
obtain the following triple.

$$
\begin{array}{l}
\hoatri{n_0 \geq 0 \andd aux_0=n_0}{\mathsf{FACT^{sa}}}
{f_2 = aux_0!}
\end{array}
$$

As expected, the program does not assign free variables from the
preconditions, that is, $\SA{\block{n_0 \geq 0 \andd aux_0=n_0}}{\mathsf{FACT^{sa}}}$.

%%%%%%%%%%%%%%%%%%%%%%%%%%%%%%%%%%%%%%%%%%%%%%%%%%%%%%%%%%%%
\section{Proving $\tsa$ is an SA translation}\label{sec:soundness}

We will now show that $\tsa$ is indeed an SA translation.

Firstly, we prove that the $\tsa$ translation preserves the
operational semantics of the original programs in the following sense:
if the translated program executes in a state where the values of the
input version of the variables coincide with the values of the
original variables in the initial state, then in the final state the
values of output versions of the variables are also equal to the
values of the original variables in the final state.

Secondly, we prove that lifting the translation function to Hoare
triples results in a sound translation, i.e., if the translated triple
is valid then the original triple must also be valid.

Finally, we will show that $\Hg$-derivability is preserved, i.e. if a
Hoare triple for an annotated program is derivable in $\Hg$, then the
translated triple is also derivable in $\Hg$.

% We will now show that $\tsa$ is indeed an SA translation.
% That is $\tsa$ conforms Definition~\ref{def:Htriples-transl-loops}. 

% Firstly, we proof that the programs produced by the $\tsa$ translation
% preserves the operational semantics of the original programs in the
% following sense:  
% If the translated program executes in a state where the values 
% of the input version of the variables coincide with the values of the
% original variables in the initial state, then also in the arrival state the values of
% output versions of the variables are equal to the values of the
% original variables in the final state.

% Secondly, we prove that 
% the translation function lifted to Hoare triples is sound, i.e., if
% the translated triple is valid then the original triple must be also
% valid.

% Finally, we will show that the $\Hg$-derivability is preserved, i.e., if
% a Hoare triple for an annotated program is derivable in $\Hg$, then the
% translated triple is also derivable in $\Hg$. 

% Before going into these proofs some lemmas are considered.

We first consider some lemmas.

\begin{lemma}\label{lem:lemma1}
Let $\Vcal \in \Var \to \mathbb{N}^+$, $s\in\states$ and $s'\in\SAstates$.
  If  $\forall x\in\Var.\, s(x)=s'(\tosa{\Vcal}(x))$, then $s' = s_0'
  \oplus\Vcal(s)$ for some $s_0'\in\SAstates$.
\end{lemma}
\begin{proof}
 Follows directly from the definitions.
\hfill $\Box$
\end{proof}

\begin{lemma}\label{lem:lemma2}
  Let $e \in {\bf Exp}$, $\phi\in\assert$,  $\Vcal \in \Var \to
  \mathbb{N}^+$, $s\in\states$ and $s'\in\SAstates$.
  \begin{enumerate}
  \item $\bsem{\tosa{\Vcal}(e)}(s'\oplus\Vcal(s)) = \bsem{e}(s)$
  \item $\bsem{\tosa{\Vcal}(\phi)}(s'\oplus\Vcal(s)) = \bsem{\phi}(s)$
  \end{enumerate}
\end{lemma}
\begin{proof}
Both proofs follow directly from Lemma~\ref{lem:lemma1}.
\hfill $\Box$
\end{proof}

\begin{lemma}\label{lem:lemma4}
Let $C \in \Acom$ and $\Vcal \in \Var \to \mathbb{N}^+$. 
If $\tsa(\Vcal, C) = (\Vcal',C')$, then for every $x\not\in\assd{C}$, $\Vcal(x)=\Vcal'(x)$.
\end{lemma}
\begin{proof}
  By induction on the structure of $C$. 
\hfill $\Box$
\end{proof}

The following lemma plays a central role in the proof of
Proposition~\ref{prop:sa_translation_final_state}, for the case of a while command.

\begin{lemma}\label{lem:while} 
%\marginpar{\todo{I not introduced (era gralha: o I sai)}}
  Let $C_t \in \Acom$, $\Vcal \in \Var \to \mathbb{N}^+$, $s_i,s_f \in \states$, $s',s_f' \in
\SAstates$ and 
\begin{align*}
&  \Vcal' = \Vcal[x \mapsto \new(\Vcal(x))  \mid x \in \assd{C_t}]  \\
&\tsa(\Vcal', C_t) = (\Vcal'',C_t') \\
& \U = [\asgn{x_{\new(\Vcal(x))}}{x_{\Vcal''(x)}} \mid x \in \assd{C_t}]
\end{align*}
If $\eval{\while{b}{\Aerase{C_t}}}{s_i}{s_f} $
and
$\eval{\while{\Vcal'(b)}{\block{\Aerase{\Tinv(C_t')};\U};\upd(\dom(\U))}}{\\
  s'
  \oplus \Vcal(s_i)}{s_f'}$ 
then $\forall x \in \Var.\, s_f(x) = s_f'(\tosa{\Vcal'}(x))$.
\end{lemma}
\begin{proof}
By induction on the derivation of the evaluation relation $\leadsto$.
Assume
\begin{align*}
&  \eval{\while{b}{\Aerase{C_t}}}{s_i}{s_f} \mbox{ and} \\
&\eval{\while{\Vcal'(b)}{\block{\Aerase{\Tinv(C_t')};\U};\upd(\dom(\U))}}{ s'
  \oplus \Vcal(s_i)}{s_f'}
\end{align*}
Two cases can occur:
\begin{itemize}
\item Case $\bsem{b}(s_i) = \Lfalse$, then $\bsem{\tosa{\Vcal'}(b)}(s' \oplus
  \Vcal'(s_i)) =\bsem{b}(s_i) = \Lfalse$  by Lemma~\ref{lem:lemma2}.
In this case we have $s_f=s_i$ and  $s_f' = s' \oplus
\Vcal'(s_i)$. Hence, $\forall x \in \Var.\, s_f'(\tosa{\Vcal'}(x)) = s_i(x) = s_f(x)$.

\item Case $\bsem{b}(s_i) = \Ltrue$, then $\bsem{\tosa{\Vcal'}(b)}(s' \oplus
  \Vcal'(s_i)) =\bsem{b}(s_i) = \Ltrue$  by Lemma~\ref{lem:lemma2}.
In this case we must have, for some $s_1\in\states$,
\begin{align}
&\eval{\Aerase{C_t}}{s_i}{s_1} \label{eq:whi1}\\
& \eval{\while{b}{\Aerase{C_t}}}{s_1}{s_f} \label{eq:wh2}
\end{align} 
and also, for some $s_0',s_1'\in\SAstates$,
\begin{align}
%  &\eval{\Aerase{\Tinv(C_t')};\U}{s_0' \oplus \Vcal'(s_i)}{s_2'} \label{eq:sa_trans_whi8}\\
  & \eval{\Aerase{\Tinv(C_t')}}{s' \oplus
    \Vcal'(s_i)}{s_0'} \label{eq:wh3} \\
& s_0' = s_2' \oplus \Vcal''(s_i)  \label{eq:wh7} \\
 & \eval{\U}{s_0'}{s_1'} \label{eq:wh4} \\
& s_1' = s_0' [ x_{\new(x_{\Vcal(x)})} \mapsto
  \bsem{x_{\Vcal''(x)}}(s_0') \mid x\in\assd{C_t}] = s_0' \oplus
  \Vcal'(s_1) \label{eq:wh5} \\
   & \eval{\while{\tosa{\Vcal'}(b)} 
                                 {\block{\Aerase{\Tinv(C_t')};\U}}}{s_1'}{s_f'} \label{eq:wh6} 
\end{align}
Note that (\ref{eq:wh7}) follows from (\ref{eq:wh3}) by
Lemma~\ref{lem:lemma4}, and that justifies  (\ref{eq:wh5}).
From (\ref{eq:wh2}), (\ref{eq:wh6}) and (\ref{eq:wh5}), by induction
hypothesis, we get $\forall x \in \Var.\, s_f(x) = s_f'(\tosa{\Vcal'}(x))$.
\end{itemize}
\hfill$\Box$
\end{proof}

We will know prove that the $\tsa$ translation
preserves the operational semantics of the original programs. i.e.,
if the translated program executes in a state where the values 
of the input version of the variables coincide with the values of the
original variables in the initial state, then in the final state the values of
output versions of the variables are also equal to the values of the
original variables in the final state.

\begin{proposition}\label{prop:sa_translation_final_state}
Let $C \in \Acom$, $\Vcal \in \Var \to \mathbb{N}^+$, $s_i,s_f \in \states$, $s',s_f' \in
\SAstates$ and $\tsa(\Vcal, C) = (\Vcal',C')$. 
If $\eval{\Aerase{C}}{s_i}{s_f} $ and $\eval{\Aerase{\Tinv(C')}}{s'
  \oplus \Vcal(s_i)}{s_f'}$, then $\forall x \in \Var.\, s_f(x) = s_f'(\tosa{\Vcal'}(x))$.
\end{proposition}
\begin{proof}
By induction on the structure of $C$.
\begin{itemize}
\item Case $C \equiv \skp$. 
The hypothesis are:
\begin{align*}
&\tsa(\Vcal,\skp) = (\Vcal,\skp) \\
& \eval{\Aerase{\skp}}{s_i}{s_i} \\ 
&\eval{\Aerase{\Tinv(\skp)}}{s' \oplus \Vcal(s_i)}{s' \oplus \Vcal(s_i)}
\end{align*}
As for every $x \in \Var$, we have $(s' \oplus
\Vcal(s_i))(\tosa{\Vcal}(x)) = s_i(x)$, we are done.
 
\item Case $C \equiv \asgn{x}{e}$. The hypothesis are:
\begin{align*}
&\tsa(\Vcal,\asgn{x}{e}) = (\Vcal[x \mapsto \nextt(\Vcal(x))], 
                      \asgn{x_{\nextt(\Vcal(x))}}{\tosa{\Vcal}(e))} \\
&\eval{\Aerase{\asgn{x}{e}}}{s_i}{s_f}  \mbox{ \ with \ }  s_f  = s_i[x
                                                          \mapsto \bsem{e}{(s_i)}]\\
&\eval{\Aerase{\Tinv(\asgn{x_{\nextt(\Vcal(x))}}{\tosa{\Vcal}(e))})}}{s'
  \oplus \Vcal(s_i)}{s_f'} \mbox{  \ with}  \\
& s_f' = (s' \oplus \Vcal(s_i))[x_{\nextt(\Vcal(x))} 
           \mapsto \bsem{\tosa{\Vcal}(e)}
                   {(s' \oplus \Vcal(s_i))}] \label{eq:sa_trans_assign_2} 
 \end{align*}
We want to prove that $\forall y\in\Var.\, s_f(y) = s_f'(\tosa{\Vcal[x \mapsto \nextt(\Vcal(x))]}(y))$.
Two cases can occur:
\begin{itemize}
\item If $y = x$, we have $s_f(y) = s_f(x) = \bsem{e}{(s_i)}$ and
$$  s_f'(\tosa{\Vcal[x \mapsto \nextt(\Vcal(x))]}(x)) =
s_f'(x_{\nextt(\Vcal(x))}) =
\bsem{\tosa{\Vcal}(e)}{(s' \oplus \Vcal(s_i))} =
\bsem{e}(s_i)
$$
\item If $y\neq x$, we have $s_f(y) = s_i(y)$ and
$$
s_f'(\tosa{\Vcal[x \mapsto \nextt(\Vcal(x))]}(y)) =
s_f'(\tosa{\Vcal}(y))  =
(s' \oplus \Vcal(s_i))(\tosa{\Vcal}(y))  =
s_i(y)
$$
\end{itemize}

\item Case $C \equiv C_1 ; C_2$. The hypothesis are:
\begin{align*}
&\tsa(\Vcal, C_1;C_2) = (\Vcal'', C_1';C_2') \\ 
 &\hspace{2.5cm}\mbox{ with }\tsa(\Vcal,C_1) = (\Vcal',C_1')  \mbox{ and } \tsa(\Vcal',C_2) =
  (\Vcal'',C_2')\\
&\eval{\Aerase{C_1;C_2}}{s_i}{s_f}\\
&\eval{\Aerase{\Tinv(C_1';C_2')}}{s \oplus \Vcal(s_i)}{s_f'}\\
\end{align*}
We must have, for some $s_0 \in \states$ and $s_0' \in \SAstates$
\begin{eqnarray}
\label{eq:sa_trans_seq1}  \eval{\Aerase{C_1}}{s_i}{s_0}\\ 
\label{eq:sa_trans_seq2} \eval{\Aerase{C_2}}{s_0}{s_f} \\ 
\label{eq:sa_trans_seq3}  \eval{\Aerase{\Tinv(C_1')}}{s \oplus \Vcal(s_i)}{s_0'} \\ 
\label{eq:sa_trans_seq4} \eval{\Aerase{\Tinv(C_2')}}{s_0'}{s_f'} 
\end{eqnarray}
From (\ref{eq:sa_trans_seq1}) and (\ref{eq:sa_trans_seq3}), by induction hypothesis, we
have $\forall x \in \Var.\, s_0(x) =
s_0'(\tosa{\Vcal'}(x))$. Therefore, by Lemma~\ref{lem:lemma1}, we have
$s_0' = s_1' \oplus \Vcal'(s_0) $ for some
$s_i'\in\SAstates$. Consequently, from  (\ref{eq:sa_trans_seq2}) and
(\ref{eq:sa_trans_seq4}), by induction hypothesis, we conclude that 
$\forall x \in \Var.\, s_f(x) = s_f'(\tosa{\Vcal''}(x))$.

\item Case $C \equiv \ifte{b}{C_t}{C_f}$.   The hypothesis are:
\begin{align*}
\tsa(\Vcal, C)  = (\supp(\Vcal',\Vcal''), \mathbf{if}\;
          {\tosa{\Vcal}(b)} \; & \mathbf{then}\; {\block{C_t';\mergee(\Vcal',\Vcal'')}} \\
 &  \mathbf{else}\: {\block{C_f';\mergee(\Vcal'',\Vcal')}} \; )\\
       \mbox{with }  \tsa(\Vcal,C_t) = (\Vcal',C_t')   
        & \mbox{ and }   \tsa(\Vcal,C_f) = (\Vcal'',C_f')
\end{align*}

\begin{tabular}{lll}
\hspace{1cm} &\multicolumn{2}{l}{$\eval{\Aerase{\ifte{b}{C_t}{C_f}}}{s_i}{s_f} $} \\
&$\langle \mathbf{if}\;  {\tosa{\Vcal}(b)} \; $ & $\mathbf{then}\;\block{C_t';\mergee(\Vcal',\Vcal'')}$ \\
& & $\mathbf{else}\;\block{C_f';\mergee(\Vcal'',\Vcal')} ,{s' \oplus \Vcal(s_i)} \rangle\!\leadsto\! {s_f'} $
\end{tabular}
\begin{itemize}
\item Case $\bsem{b}(s_i) = \Ltrue$, then $\bsem{\tosa{\Vcal}(b)}(s' \oplus
  \Vcal(s_i)) = \Ltrue$  by Lemma~\ref{lem:lemma2}.
Therefore one must have, for some $s_0'\in\SAstates$,
% \begin{align}
%   & \eval{\Aerase{C_t}}{s_i}{s_f} \label{eq:sa_trans_if1}\\
%   & \eval{\Aerase{\Tinv(C_t)}}{s' \oplus \Vcal(s_i)} {s_0'} \label{eq:sa_trans_if2} \\
%   & \eval{\mergee(\Vcal',\Vcal'')} {s_0'}{s_f'} \label{eq:sa_trans_if3}
% \end{align}
\begin{eqnarray}
 \label{eq:sa_trans_if1} &\eval{\Aerase{C_t}}{s_i}{s_f} \\
 \label{eq:sa_trans_if2} &\eval{\Aerase{\Tinv(C_t')}}{s' \oplus
  \Vcal(s_i)} {s_0'} \\
 \label{eq:sa_trans_if3} &\eval{\mergee(\Vcal',\Vcal'')} {s_0'}{s_f'} 
\end{eqnarray}
From (\ref{eq:sa_trans_if1}) and (\ref{eq:sa_trans_if2}), by induction 
hypothesis we have that 
\begin{eqnarray}
\forall x \in \Var.\, s_f(x) = s_0'(\tosa{\Vcal'}(x)) \label{eq:sa_trans_if4}
\end{eqnarray}
$\mergee\ (\Vcal',\Vcal'') = [\asgn{x_{\Vcal''(x)}}{x_{\Vcal'(x)}} \mid
                                          x \in \V \wedge \Vcal'(x)
                                          \prec \Vcal''(x)] $ so, 
$s_f' = s_0' [ x_{\Vcal''(x)} \mapsto \bsem{x_{\Vcal'(x)}}(s_0') \mid x \in \Var
\wedge \Vcal'(x) \prec \Vcal''(x) ]$. Moreover,
\begin{align*}
  \supp\ (\Vcal',\Vcal'')(x) =  \left\{ 
        \begin{array}{l l}
          \Vcal'(x) & \caseif\ \Vcal''(x) \prec \Vcal'(x) \\
          \Vcal''(x) & \otherwise
        \end{array} \right.\
\end{align*}
We will now  prove that $\forall x\in\Var.\, s_f(x) =
s_f'(\tosa{\supp(\Vcal',\Vcal'')}(x))$. Let $x\in\Var$.
\begin{itemize}
\item If $\Vcal'(x) \prec \Vcal''(x)$, then
$s_f'(\tosa{\supp(\Vcal',\Vcal'')}(x)) = s_f'(\tosa{\Vcal''}(x)) =
\bsem{x_{\Vcal'(x)}}(s_0')= s_0'(\tosa{\Vcal'}(x)) = s_f(x)$ by (\ref{eq:sa_trans_if4}).
\item  If $\Vcal'(x) \not\prec \Vcal''(x)$, then
  $s_f'(\tosa{\supp(\Vcal',\Vcal'')}(x)) = s_f'(\tosa{\Vcal'}(x)) = s_0'(\tosa{\Vcal'}(x)) = s_f(x)$ by (\ref{eq:sa_trans_if4}).
\end{itemize}

\item Case $\bsem{b}(s_i) = \Lfalse$. Analogous to the previous case.
\end{itemize}

\item Case $C \equiv \whileinv{b}{\theta}{C_t}$.
The hypothesis are:
      \begin{align}
        \tsa(\Vcal, \whileinv{b}{\theta}{C_y}) 
          &= (\Vcal''', \forinv{\I}
                              {\tosa{\Vcal'}(b)}
                              {\U}
                              {\tosa{\Vcal'}(\theta)}
                              {C_t'}
            )\\
          \mbox{with } &\I = [x_{\new(\Vcal(x))} 
                      | x \in \assd{C_t}] \label{eq:sa_trans_whi1}\\
                &\Vcal' = \Vcal[x \mapsto \new(\Vcal(x)) 
                               \mid x \in \assd{C_t}] \label{eq:sa_trans_whi2}\\
                &(\Vcal'',C_t') = \tsa(\Vcal',C_t) \label{eq:sa_trans_whi3}\\
                &\U = [\asgn{x_{\new(\Vcal(x))}}
                            {x_{\Vcal''(x)}}
                      \mid x \in \assd{C_t}] \label{eq:sa_trans_whi4}\\
                &\Vcal''' = \Vcal''[x \mapsto \jump(l) 
                                   \mid x_l \in
                  \dom(\U)]\label{eq:sa_trans_whi5} 
      \end{align}
 \begin{align}
 & \eval{\while{b}{\Aerase{C_t}}}{s_i}{s_f} \label{eq:sa_trans_whi6} \\
& \eval{\I;\while{\Vcal'(b)}{\block{\Aerase{\Tinv(C_t')};\U};\upd(\dom(\U))}}{s' \oplus \Vcal(s_i)}{s_f'} \label{eq:sa_trans_whi7} 
 \end{align}
From (\ref{eq:sa_trans_whi1}), (\ref{eq:sa_trans_whi2}) and (\ref{eq:sa_trans_whi7}), we must have for some $s_0', s_1' \in \SAstates$ that: 
      \begin{align}
        &\eval{\I}{s' \oplus s_i}{s_1'} \nonumber \\
        &\eval{\while{\Vcal'(b)}{\block{\Aerase{\Tinv(C_t')} ; \U} ;
        \upd(\dom(\U))}}{s_1'}{s_f'}   \label{eq:sa_trans_whi13} \\
       & s_1' = s_0' \oplus \Vcal'(s_i)  \nonumber
      \end{align}
There are the following cases to consider:
\begin{itemize}
\item Case $\bsem{b}{(s_i)} = \Lfalse$, then $s_f = s_i$ and
  $\bsem{\tosa{\Vcal'}(b)}{(s' \oplus \Vcal'(s_i))} = \Lfalse$, by Lemma~\ref{lem:lemma2}.
  Therefore one must have
\begin{align*}
            &\eval{\while{b}{\Aerase{C_t}}}{s_i}{s_i}\\
            &\eval{\while{\Vcal'(b)}
                                 {\block{\Aerase{\Tinv(C_t')} ; \U}}}
                          {s_1'}
                          {s_1'}  \\
           & \eval{\upd(\dom(\U))}{s_1'}{s_f'}
          \end{align*}
Moreover, we know that $\upd(\dom(\U)) = [\asgn{x_{\jump(l)}}{x_l} \mid
x_l \in \dom(\U)]$ so
\begin{align}
 &\eval{\upd(\dom(\U))}{s_0' \oplus \Vcal'(s_i)}
  {s_f'} \label{eq:sa_trans_whi6} \\
& s_f' =(s_0' \oplus \Vcal'(s_i))[x_{\jump(l)} \mapsto \bsem{x_l}{(s_0' \oplus \Vcal'(s_i))}
                                            \mid x_l \in \dom(\U)]
\end{align}
We now prove that $\forall x\in\Var.\, s_f(x) =
s_f'(\tosa{\Vcal'}(x))$.
\begin{itemize}
\item If $y \not\in \assd{C_t}$, then
$ s_f'(\tosa{\Vcal'''}(y)) = s_f'(\tosa{\Vcal''}(y)) =
s_f'(\tosa{\Vcal'}(y)) = s_i(y) $, using Lemma~\ref{lem:lemma4}.          

\item If $x \in \assd{C_t}$, then 
$s_f'(\tosa{\Vcal'''}(x)) = s_f'(x_{\jump(l)})$ for some $x_l \in \dom(\U)$ and
$s_f'(x_{\jump(l)}) = \bsem{x_l}{(s_0' \oplus \Vcal'(s_i))}
= \bsem{\tosa{\Vcal'}(x)}{(s_0' \oplus \Vcal'(s_i))}
= \bsem{x}(s_i)
= s_i(x)$, using Lemma~\ref{lem:lemma1}.     
\end{itemize}

\item Case $\bsem{b}{(s_i)} = \Ltrue$, then
 $\bsem{\tosa{\Vcal'}(b)}{(s_0' \oplus \Vcal'(s_i))} = \Ltrue$,  by Lemma~\ref{lem:lemma2}.
From (\ref{eq:sa_trans_whi6}), we must have, for some $s_2\in\states$,
          \begin{align}
            &\eval{\Aerase{C_t}}{s_i}{s_2} \label{eq:sa_trans_whi7}\\
            & \eval{\while{b}
                          {\Aerase{C_t}}}{s_2}
                          {s_f} \label{eq:sa_trans_whi8}
          \end{align} 
 and also, from (\ref{eq:sa_trans_whi13}), for some $s_2',s_3',s_4'\in\SAstates$,
 \begin{align}
%  &\eval{\Aerase{\Tinv(C_t')};\U}{s_0' \oplus \Vcal'(s_i)}{s_2'} \label{eq:sa_trans_whi8}\\
  & \eval{\Aerase{\Tinv(C_t')}}{s_0' \oplus \Vcal'(s_i)}{s_4'} \label{eq:sa_trans_whi9} \\
 & \eval{\U}{s_4'}{s_2'} \label{eq:sa_trans_whi10} \\
& s_2' = s_4' [ x_{\new(x_{\Vcal(x)})} \mapsto
  \bsem{x_{\Vcal''(x)}}(s_4') \mid x\in\assd{C_t}] \label{eq:sa_trans_whi11} \\
   & \eval{\while{\tosa{\Vcal'}(b)} 
                                 {\block{\Aerase{\Tinv(C_t')};\U}}}{s_2'}{s_3'} \label{eq:sa_trans_whi12} \\
         & \eval{\upd(\dom(\U))}{s_3'}{s_f'}
\end{align}
From (\ref {eq:sa_trans_whi7}), (\ref{eq:sa_trans_whi9}) and
(\ref{eq:sa_trans_whi3}), by induction hypothesis, we have 
that $\forall x\in\Var.\, s_2(x) = s_4'(\tosa{\Vcal''}(x))$.
Using  this fact, (\ref{eq:sa_trans_whi11}) and
(\ref{eq:sa_trans_whi2}), we can conclude that $s_2' = s_4' 
\oplus \Vcal'(s_2)$. Because of this, (\ref{eq:sa_trans_whi8}) and
(\ref{eq:sa_trans_whi12}), we get by Lemma~\ref{lem:while}  
\begin{align}
  \forall x\in\Var.\, s_f(x) = s_3'(\tosa{\Vcal'}(x))  \label{eq:7}
\end{align}

$\upd(\dom(\U)) = [ \asgn{x_{\jump(l)}}{x_l} \mid x_l \in\dom(\U) ]$
and $\eval{\upd(\dom(\U)) }{s_3'}{s_f'}$, so $s_f'= s_3' [ x_{\jump(l)}
\mapsto \bsem{x_l}(s_3') \mid x_l \in \dom(\U) ]$.
We will now  prove that $\forall x\in\Var.\, s_f(x) =
s_f'(\tosa{\Vcal'''}(x))$. 
\begin{itemize}
\item If $y\not\in \assd{C_t}$, then $s_f'(\tosa{\Vcal'''}(y)) =
  s_f'(\tosa{\Vcal''}(y)) = s_3'(\tosa{\Vcal'}(y)) = s_f(y)$, using
  Lemma~\ref{lem:lemma4} and (\ref{eq:7}).
\item If $x\in \assd{C_t}$, then $s_f'(\tosa{\Vcal'''}(x)) =
  s_f'(x_{\jump(l)})$ for some $x_l=x_{\new(\Vcal(c)}\in\dom(\U)$, and
  $s_f'(x_{\jump(l)}) = \bsem{x_l}(s_3') = s_f(x)$, using (\ref{eq:7}).
\end{itemize}

\end{itemize}

\end{itemize}
\hfill$\Box$
\end{proof}

We now prove that the translation function is sound, i.e., if
the translated triple is valid then the original triple must also be
valid. We need the following lemma.

\begin{lemma}\label{lem:lemma5}
 Let $C \in \Acom$, $\Vcal \in \Var \to \mathbb{N}^+$, and $s_i,
 s_f\in\states$. 
If  $\eval{\Aerase{C}}{s_i}{s_f}$  and $\tsa(\Vcal, C) = (\Vcal',C')$,
then $\eval{\Aerase{\Tinv(C')}}{s_1'\oplus\Vcal(s_i)}{s_2'\oplus \Vcal'(s_f)}$, for some
$s_1',s_2'\in \SAstates$.
\end{lemma}
\begin{proof}
By induction  on the structure of $C$ using
Proposition~\ref{prop:sa_translation_final_state}. %\todo{falta fazer esta prova}
\hfill $\Box$
\end{proof}

\begin{proposition}\label{prop:TsaSound}
  Let $C\in \Acom$, $\phi,\psi\in\assert$, $\Vcal \in \Var \to
  \mathbb{N}^+$ and $\tsa(\Vcal,C) = (\Vcal',C')$.
If $\:\models \hoatri{\tosa{\Vcal}(\phi)}{\Aerase{\Tinv(C')}}{\tosa{\Vcal'}(\psi)} $,
then $\:\models \hoatri{\phi}{\Aerase{C}}{\psi}$.
\end{proposition}
\begin{proof}
 Let $\tsa(\Vcal,C) = (\Vcal',C')$ and
 \begin{align}
   \models
   \hoatri{\tosa{\Vcal}(\phi)}{\Aerase{\Tinv(C')}}{\tosa{\Vcal'}(\psi)}  \label{eq:p1}
 \end{align}
We want to prove that $\models \hoatri{\phi}{\Aerase{C}}{\psi}$, so
assume, for some $s_i,s_f\in\states$, that
\begin{align}
  \bsem{\phi}(s_i)=\Ltrue \label{eq:p2}\\
  \eval{\Aerase{C}}{s_i}{s_f} \label{eq:p3}
\end{align}
From (\ref{eq:p3}), by Lemma~\ref{lem:lemma5}, we have for some $s_1',s_2'\in\SAstates$
\begin{align}
%  \eval{\Aerase{\Tinv(C')}}{s'\oplus\Vcal(s_i)}{s_f'} \label{eq:p4}
  \eval{\Aerase{\Tinv(C')}}{s_1'\oplus\Vcal(s_i)}{s_2'\oplus\Vcal'(s_f)} \label{eq:p4}
\end{align}
From (\ref{eq:p2}),  by Lemma~\ref{lem:lemma2}, we have 
$\bsem{\tosa{\Vcal}(\phi)}(s_1'\oplus\Vcal(s_i))=\Ltrue$.
Thus, since we have  (\ref{eq:p1}) and (\ref{eq:p4}), we get 
$\bsem{\tosa{\Vcal'}(\psi)}(s_2'\oplus\Vcal'(s_f)) = \Ltrue$ and, by
Lemma~\ref{lem:lemma2}, it follows that $\bsem{\psi}(s_f)=\Ltrue$.
Hence  $\:\models \hoatri{\phi}{\Aerase{C}}{\psi}$ holds.
%
% for some $s'\in\SAstates$. 
% Thus, using (\ref{eq:p4}), we get $\bsem{\psi'}(s_f')=\Ltrue$. 
% \marginpar{\todo{psi' not introduced\\ from (36) and (38)... ->\\
% from (37) and (38)...} A prova foi simplificada.}
% From  (\ref{eq:p2}) and (\ref{eq:p4}), by
% Proposition~\ref{prop:sa_translation_final_state}, we have $\forall
% x\in\Var.\, s_f(x)=s_f'(\tosa{\Vcal'}(x))$. So,
% $\bsem{\psi}(s_f)=\Ltrue$, because $\bsem{\psi'}(s_f') =
% \bsem{\tosa{\Vcal'}(\psi)}(s_f') =\bsem{\psi}(s_f)$. 
% Hence  $\:\models \hoatri{\phi}{\Aerase{C}}{\psi}$ holds.
\hfill$\Box$
\end{proof}

Finally we will show that the translation $\tsa$ preserves  $\Hg$-derivations, i.e., if
a Hoare triple for an annotated program is derivable in $\Hg$, then the
translated triple is also derivable in $\Hg$. 
Again we start by proving some auxiliary lemmas.

\begin{lemma}\label{lem:merge_sup}
Let $\Vcal,\Vcal' \in \Var \to \mathbb{N}^+$ and $\psi \in \assert$. 
The following derivations hold
\begin{enumerate}
\item $\infHg \hoatri{\tosa{\Vcal}(\psi)}
                {\mergee(\Vcal,\Vcal')}
                {\tosa{\supp(\Vcal,\Vcal')}(\psi)}$
\item $\infHg \hoatri{\tosa{\Vcal}(\psi)}
                {\mergee(\Vcal',\Vcal)}
                {\tosa{\supp(\Vcal,\Vcal')}(\psi)}$            
\end{enumerate}
\end{lemma}
\begin{proof} 
1. 
We have
$\mergee\ (\Vcal,\Vcal') = [\asgn{x_{\Vcal'(x)}}{x_{\Vcal(x)}} \mid
                            x \in \V \wedge \Vcal(x) \prec \Vcal'(x)]
$
and
$$\supp\ (\Vcal,\Vcal')(x) =  \left\{ 
        \begin{array}{l l}
          \Vcal(x) & \caseif\ \Vcal'(x) \prec \Vcal(x) \\
          \Vcal'(x) & \otherwise
        \end{array} \right.
$$
$\infHg \hoatri{\tosa{\Vcal}(\psi)}
                {\mergee(\Vcal,\Vcal')}
                {\tosa{\supp(\Vcal,\Vcal')}(\psi)}$ 
follows from successively applying the (assign) and (seq) rules, using
as precondition the postcondition with the substitution
$[\subst{x_{\Vcal'(x)}}{x_{\Vcal(x)}}]$ for each assigment $\asgn{x_{\Vcal'(x)}}{x_{\Vcal(x)}}$  of the
renaming sequence, since 
$\tosa{\supp(\Vcal,\Vcal')}(\psi)
      [\subst{x_{\Vcal'(x)}}{x_{\Vcal(x)}}  \mid
       x \in \V \wedge \Vcal(x) \prec \Vcal'(x)] = \tosa{\Vcal}(\psi)$.

2. Similar.
\hfill$\Box$
\end{proof}

\begin{lemma}\label{lem:i_pre_post}
Let $\I \in \rnm$ and $\phi \in \assert$. 
\[
\infHg \hoatri{\phi}{\I}{\I^{-1}(\phi)}
\]
\end{lemma}
\begin{proof}
Follows directly from Lemma~\ref{lemma:rnm}.
\hfill$\Box$
\end{proof}

\begin{lemma}\label{lem:u_post}
Let $\Vcal \in \Var \to \mathbb{N}^+$, $C \in \Acom$,
$\Vcal' = \Vcal[x \mapsto
\new(\Vcal(x)) \mid x \in \assd{C}]$, $\tsa(\Vcal',C) = (\Vcal'',C')$
and $\Uc = [\asgn{x_{\new(\Vcal(x))}}{x_{\Vcal''(x)}} \mid x \in \assd{C}]$.
\[
\infHg \hoatri{\tosa{\Vcal''}(\theta)}
              {\Uc}
              {\tosa{\Vcal'}(\theta)}
\]
\end{lemma}
\begin{proof}
$\infHg \hoatri{\tosa{\Vcal''}(\theta)}
              {\Uc}
              {\tosa{\Vcal'}(\theta)}$
follows from successively applying the (assign) and (seq) rules, using
as precondition the postcondition with the substitution
$[\subst{x_{\new(\Vcal(x))}}{{x_{\Vcal''(x)}}} ]$ for each assigment $\asgn{x_{\new(\Vcal(x))}}{x_{\Vcal''(x)}}$  of the
renaming sequence, since
\begin{align*}
& (\tosa{\Vcal'}(\theta))
[\subst{x_{\new(\Vcal(x))}}{{x_{\Vcal''(x)}}} \mid x \in \assd{C}]  \\
= \: &
(\tosa{\Vcal[x \mapsto \new(\Vcal(x)) \mid x \in \assd{C}]}(\theta))
[\subst{x_{\new(\Vcal(x))}}{{x_{\Vcal''(x)}}} \mid x \in \assd{C}] \\
= \: &\tosa{\Vcal''}(\theta)
\end{align*}
\hfill$\Box$
\end{proof}

\begin{lemma}\label{lem:upd_post}
Let $\Vcal \in \Var \to \mathbb{N}^+$, $C \in \Acom$, 
$\Vcal' = \Vcal[x \mapsto \new(\Vcal(x)) 
                     \mid x \in \assd{C}]$,
$\tsa(\Vcal',C) = (\Vcal'',C')$,
$\Uc = [\asgn{x_{\new(\Vcal(x))}}
                  {x_{\Vcal''(x)}}
            \mid x \in \assd{C}]$ and 
$\Vcal''' = \Vcal''[x \mapsto \jump(l) 
                         \mid x_l \in
        \dom(\Uc)]$.
\[
\infHg \hoatri{\tosa{\Vcal'}(\psi)}
              {\upd(\dom(\Uc))}
              {\tosa{\Vcal'''}(\psi)}
\]
\end{lemma}
\begin{proof}
Remember that $\upd\ (\dom(\Uc)) 
= [\asgn{x_{\jump(l)}}{x_l} \mid x_l \in \dom(\Uc) ]$.

$\infHg \hoatri{\tosa{\Vcal'}(\psi)}
              {\upd(\dom(\Uc))}
              {\tosa{\Vcal'''}(\psi)}$
follows from successively applying the (assign) and (seq) rules, using
as precondition the postcondition with the substitution
$[\subst{x_{\jump(l)}}{x_l} ]$ for each assigment $\asgn{x_{\jump(l)}}{x_l}$  of the
renaming sequence, since
\begin{align*}
&  (\tosa{\Vcal'''}
  (\psi)[\subst{x_{\jump(l)}}{x_l}
        \mid x_l \in \dom(\Uc)]) \\
= \: & (\tosa{\Vcal''[x \mapsto \jump(l) \mid x_l 
                                     \in \dom(\Uc)]}
  (\psi)[\subst{x_{\jump(l)}}{x_l}
        \mid x_l \in \dom(\Uc)]) \\
= \: &\Vcal'(\psi)
\end{align*}
\hfill$\Box$
\end{proof}

\begin{proposition}\label{prop:TsaComplete}
Let $C \in \Acom$, $\phi,\psi \in \assert$, $\Vcal \in \Var \to
\mathbb{N}^+$ and $\tsa(\Vcal, C) = (\Vcal',C')$. 
If $\infHg \hoatri{\phi}{C}{\psi}$, then 
$\infHg \hoatri{\tosa{\Vcal}(\phi)}{\Tinv(C')}{\tosa{\Vcal'}(\psi)}$.
\end{proposition}
\begin{proof}
By induction on the structure of $\infHg \hoatri{\phi}{C}{\psi}$. 
%For all cases let us assume some predefined $\Vcal \in \Var \to \mathbb{N}^+$.
\begin{itemize}
\item Assume the last step is:
$$
\begin{prooftree}
  \justifies
  \hoatri{\phi}{\skp}{\psi} \using \quad \mbox{with $\phi \to \psi$}
\end{prooftree}
$$
We have $\tsa(\Vcal, \skp) = (\Vcal,\skp)$.
As $\models \phi \to \psi$ we have
$\models \tosa{\Vcal}(\phi) \to \tosa{\Vcal}(\psi)$.
So $\infHg \hoatri{\tosa{\Vcal}(\phi)}
                                                {\skp}
                                                {\tosa{\Vcal}(\psi)}$
by applying the (skip) rule.

\item Assume the last step is:
$$
\begin{prooftree}
  \justifies
  \hoatri{\phi}{\asgn{x}{e}}{\psi} \using \quad 
    \mbox{with $\phi \to \psi[\subst{x}{e}]$}
\end{prooftree}
$$
We have $\tsa(\Vcal, \asgn{x}{e}) 
  = (\Vcal[x \mapsto \nextt(\Vcal(x))], 
     \asgn{x_{\nextt(\Vcal(x))}}{\tosa{\Vcal}(e)}
    )
$.
Since $\models \phi \to \psi[\subst{x}{e}]$, it follows that
$\models \tosa{\Vcal}(\phi \to \psi[\subst{x}{e}]) $.
Moreover, 
\begin{align*}
&\tosa{\Vcal}(\phi) \to 
              \tosa{\Vcal[x \mapsto \nextt(\Vcal(x))]}(\psi)
                  [\subst{x_{\nextt(\Vcal(x))}}{\tosa{\Vcal}(e)}]\\
= \: &\tosa{\Vcal}(\phi) \to 
              \tosa{\Vcal}(\psi[\subst{x}{x_{\nextt(\Vcal(x))}}])
                  [\subst{x_{\nextt(\Vcal(x))}}{\tosa{\Vcal}(e)}]\\
= \: &\tosa{\Vcal}(\phi) \to 
              \tosa{\Vcal}(\psi[\subst{x}{x_{\nextt(\Vcal(x))}}]
              [\subst{x_{\nextt(\Vcal(x))}}{e}]) 
              \quad \mbox{, because }x_{\nextt(\Vcal(x))} \notin \FV{\psi}\\
= \: &\tosa{\Vcal}(\phi) \to 
              \tosa{\Vcal}(\psi[\subst{x}{e}]) \\
= \: &\tosa{\Vcal}(\phi \to \psi[\subst{x}{e}])
\end{align*}
Hence
$\infHg \hoatri{\tosa{\Vcal}(\phi)}
         {\asgn{x_{\nextt(\Vcal(x))}}{\tosa{\Vcal}(e)}}
         {\tosa{\Vcal[x \mapsto \nextt(\Vcal(x))]}(\psi)}$ 
follows by the (assign) rule.

\item Assume the last step is:
$$
\begin{prooftree}
  \hoatri{\phi}{C_1}{\theta} \quad \hoatri{\theta}{C_2}{\psi}
  \justifies
  \hoatri{\phi}{C_1 \sep C_2}{\psi}
\end{prooftree}
$$
\begin{align*}
\mbox{We have } \tsa(\Vcal, C_1 \sep C_2) =& (\Vcal'',C_{1}'; C_{2}')\\
\mbox{with }&\tsa(\Vcal, C_1) = (\Vcal',C_{1}') \mbox{ and }
     \tsa(\Vcal', C_2) = (\Vcal'',C_{2}').
\end{align*}
By induction hypothesis, we have
$\infHg
\hoatri{\tosa{\Vcal}(\phi)}{\Tinv(C_1')}{\tosa{\Vcal'}(\theta)}$ and
also  $\infHg \hoatri{\tosa{\Vcal'}(\theta)}{\Tinv(C_2')}{\tosa{\Vcal''}(\psi)}$.
Hence, applying the (seq) rule, we get 
\[\infHg \hoatri{\tosa{\Vcal}(\phi)}
                                               {\Tinv(C_1' \sep C_2')}
                                               {\tosa{\Vcal''}(\psi)}
\]

\item Assume the last step is:
$$
\begin{prooftree}
  \hoatri{\phi \andd \boolemb{b}}{C_t}{\psi} \quad 
    \hoatri{\phi \andd \negg \boolemb{b}}{C_f}{\psi}
  \justifies
  \hoatri{\phi}{\ifte{b}{C_t}{C_f}}{\psi}
\end{prooftree}
$$
\begin{align*}
\mbox{We have \ \ } \tsa(\Vcal, C)  = (\supp(\Vcal',\Vcal''), \mathbf{if}\;
          {\tosa{\Vcal}(b)} \; & \mathbf{then}\; {\block{C_t';\mergee(\Vcal',\Vcal'')}} \\
 &  \mathbf{else}\: {\block{C_f';\mergee(\Vcal'',\Vcal')}} \; )\\
       \mbox{with }  \tsa(\Vcal,C_t) = (\Vcal',C_t')   
        & \mbox{ and }   \tsa(\Vcal,C_f) = (\Vcal'',C_f')
\end{align*}
\begin{itemize}
\item By induction hypothesis we have that 
$\infHg \hoatri{\tosa{\Vcal}(\phi \andd
  \boolemb{b})}{\Tinv(C_t')}{\tosa{\Vcal'}(\psi)}$. 
By Lemma~\ref{lem:merge_sup} we have  
$\infHg \hoatri{\tosa{\Vcal'}(\psi)}
               {\mergee(\Vcal',\Vcal'')}
               {\tosa{\supp(\Vcal',\Vcal'')}(\psi)}$. So, applying
               rule (seq), we get
\begin{align}
\infHg \hoatri{\tosa{\Vcal}(\phi \andd \boolemb{b})}
              {\Tinv(C_t')\sep \mergee(\Vcal',\Vcal'')}
              {\tosa{\supp(\Vcal',\Vcal'')}(\psi)}\label{eq:sa_trans2_if3}
\end{align}

\item  By induction hypothesis we have that 
$\infHg \hoatri{\tosa{\Vcal}(\phi \andd
  \neg \boolemb{b})}{\Tinv(C_f')}{\tosa{\Vcal'}(\psi)}$. 
By Lemma~\ref{lem:merge_sup} we have  
$\infHg \hoatri{\tosa{\Vcal'}(\psi)}
               {\mergee(\Vcal'',\Vcal')}
               {\tosa{\supp(\Vcal',\Vcal'')}(\psi)}$. So, applying
               rule (seq), we get
\begin{align}
\infHg \hoatri{\tosa{\Vcal}(\phi \andd \negg \boolemb{b})}
              {\Tinv(C_f')\sep \mergee(\Vcal'',\Vcal')}
              {\tosa{\supp(\Vcal',\Vcal'')}(\psi)}\label{eq:sa_trans2_if4}
\end{align}
\end{itemize}
Finally, from (\ref{eq:sa_trans2_if3}) and (\ref{eq:sa_trans2_if4}), by
rule (if) we get 
\begin{align*}
\infHg \hoatri{\tosa{\Vcal}(\phi)}
               { \mathbf{if}\;
          {\tosa{\Vcal}(b)} \; & \mathbf{then}\; {\block{\Tinv(C_t');\mergee(\Vcal',\Vcal'')}} \\
 &  \mathbf{else}\: {\block{\Tinv(C_f');\mergee(\Vcal'',\Vcal')}} }
                {\tosa{\supp(\Vcal',\Vcal'')}(\psi))}
\end{align*}
 
\item Assume the last step is:
$$
\begin{prooftree}
  \hoatri{\theta \andd \boolemb{b}}{C}{\theta} 
  \justifies
  \hoatri{\phi}{\whileinv{b}{\theta}{C}}{\psi} \using 
    \mbox{with $\phi \impl \theta\ \mbox{ and } 
              \theta \andd \negg b \impl \psi $}
\end{prooftree}
$$
We have
\begin{align}
\tsa(\Vcal, \whileinv{b}{\theta}{C}) 
    &= (\Vcal''', \forinv{\Ic}
                        {\tosa{\Vcal'}(b)}
                        {\U}
                        {\tosa{\Vcal'}(\theta)}
                        {C'}; \upd(\dom(\Uc))
      )\nonumber\\
    \mbox{with } & \Ic = [\asgn{x_{\new(\Vcal(x))}}
                               {x_{\Vcal(x)}}
                          \mid x \in \assd{C}] \nonumber\\
          &\Vcal' = \Vcal[x \mapsto \new(\Vcal(x)) 
                         \mid x \in \assd{C}] \nonumber\\
          &(\Vcal'',C') = \tsa(\Vcal',C) \nonumber\\
          &\U = [\asgn{x_{\new(\Vcal(x))}}
                      {x_{\Vcal''(x)}}
                \mid x \in \assd{C}] \nonumber\\
          &\Vcal''' = \Vcal''[x \mapsto \jump(l) 
                             \mid x_l \in
            \dom(\U)]\nonumber
\end{align}
\begin{align}
\mbox{and  \ \ } & \Tinv(\forinv{\Ic}
                        {\tosa{\Vcal'}(b)}
                        {\U}
                        {\tosa{\Vcal'}(\theta)}
                        {C'} \sep \upd(\dom(\U)))  \qquad\qquad\:  \nonumber\\
& = \I \sep \whileinv{\tosa{\Vcal'}(b)}
                   {\tosa{\Vcal'}(\theta)}
                   {\block{\Tinv(C') \sep \U}}
     \sep \upd(\dom(\U))\nonumber
\end{align}

We must prove that\\
$\infHg \hoatri{\tosa{\Vcal}(\phi)}
               {\Ic \sep \whileinv{\tosa{\Vcal'}(b)}
                                 {\tosa{\Vcal'}(\theta)}
                                 {\block{\Tinv(C') \sep \U}}
                    \sep \upd(\dom(\U))}
                {\tosa{\Vcal'''}(\psi)} $
%We show next a possible derivation. 
This follows from applying rule (seq) twice, to the following premisses:
\begin{align}
  & \infHg \hoatri{\tosa{\Vcal}(\phi)}
                {\Ic}
                {\Ic^{-1}(\tosa{\Vcal}(\phi))} \label{eq:d1} \\
& \infHg \hoatri{\Ic^{-1}(\tosa{\Vcal}(\phi))}
                  {\whileinv{\tosa{\Vcal'}(b)}
                            {\tosa{\Vcal'}(\theta)}
                            {\block{\Tinv(C') \sep \U}}}
                {\tosa{\Vcal'}(\psi)}  \label{eq:d2} \\
&\infHg \hoatri{\tosa{\Vcal'}(\psi)}
                  {\upd(\dom(\U))}
                  {\tosa{\Vcal'''}(\psi)}  \label{eq:d3} 
\end{align}
We have that (\ref{eq:d1}) follows from Lemma~\ref{lem:i_pre_post},
and  (\ref{eq:d3}) follows from Lemma~\ref{lem:upd_post}.
We will now prove  (\ref{eq:d2}). 
By induction hypotesis, we have
$\infHg \hoatri{\tosa{\Vcal'}(\theta) \wedge \tosa{\Vcal'}(b)}
                    {\Tinv(C')}
                    {\tosa{\Vcal''}(\theta)}$. 
Moreover, by Lemma~\ref{lem:u_post},  $\infHg \hoatri{\tosa{\Vcal''}(\theta)}
                      {\U}
                      {\tosa{\Vcal'}(\theta)}$. Thus, by rule (seq),
$
\infHg   \hoatri{\tosa{\Vcal'}(\theta) \wedge \tosa{\Vcal'}(b)}
                    {\Tinv(C') \sep \U}
                    {\tosa{\Vcal'}(\theta)}   
$.

Since 
      $(\I^{-1}(\tosa{\Vcal}(\phi)) \to \tosa{\Vcal'}(\theta))$ and
      $(\tosa{\Vcal'}(\theta) \wedge \neg \tosa{\Vcal'}(b) \to
      \tosa{\Vcal'}(\psi))$ both hold, we can now apply rule (while), and obtain
\[
\infHg \hoatri{\I^{-1}(\tosa{\Vcal}(\phi))}
                  {\whileinv{\tosa{\Vcal'}(b)}
                            {\tosa{\Vcal'}(\theta)}
                            {\block{\Tinv(C') \sep \U}}}
                {\tosa{\Vcal'}(\psi)} 
\]
\end{itemize}
\hfill$\Box$
\end{proof}

%We are now ready to prove the main result of this section.

It is now immediate that $\tsa$ conforms th Definition~\ref{def:Htriples-transl-loops}.

\begin{proposition}\label{prop:tsaSAtranlation}
  The $\tsa$ function of Figure~\ref{fig:sa-translation}  is an SA translation.
\end{proposition}
\begin{proof}
 Let $C\in \Acom$, $\phi,\psi\in\assert$, $\Vcal \in \Var \to
  \mathbb{N}^+$, $\tsa(\Vcal,C) = (\Vcal',C')$ and $\tsa(\phi,C,\psi)
  = (\tosa{\Vcal}(\phi),C',\tosa{\Vcal'}(\psi))$. 
We want to prove that
\begin{enumerate}
\item If $\:\models \hoatri{\tosa{\Vcal}(\phi)}{\Aerase{\Tinv(C')}}{\tosa{\Vcal'}(\psi)} $,
then $\:\models \hoatri{\phi}{\Aerase{C}}{\psi}$.
\item  If  $\: \infHg \hoatri{\phi}{C}{\psi}$, then
 $\: \infHg \hoatri{\tosa{\Vcal}(\phi)}{\Tinv(C')}{\tosa{\Vcal'}(\psi)}$.
\end{enumerate}
1. follows directly from Proposition~\ref{prop:TsaSound}.\\
2. follows directly from  Proposition~\ref{prop:TsaComplete}.
\hfill$\Box$
\end{proof}

%%%%%%%%%%%%%%%%%%%%%%%%%%%%%%%%%%%%%%%%%%%%%%%%%%%%%%%%%%%%
\section{Conclusion}\label{sec:conclusion}

We have proposed a translation of programs annotated with loop invariants into a dynamic single-assignment form. The translation extends to Hoare triples and is adequate for program verification using the efficient VCGen presented in~\cite{PintoJS:sinapv}. Together with that VCGen and the corresponding program logic for single-assignment programs, the translation is part of a workflow for the deductive verification of imperative programs in a way that is efficient and allows for adaptation, since the logic is adaptation-complete. We remark that the workflow, also described in that paper, does not depend on this specific translation, but instead defines semantic requirements that a translation should comply to. The core result of the present report, proved in detail here, is precisely that the specific translation introduced here conforms to those requirements.

%%%%%%%%%%%%%%%%%%%%%%%%%%%%%%%%%%%%%%%%%%%%%%%%%%%%%%%%%%%%
\bibliographystyle{plain}
\bibliography{}

\begin{thebibliography}{1}

\bibitem{AptKR:tenyhlsone}
Krzysztof~R. Apt.
\newblock Ten years of {Hoare}'s logic: A survey - part 1.
\newblock {\em ACM Trans. Program. Lang. Syst.}, 3(4):431--483, 1981.

\bibitem{CookSA:soucaspv}
Stephen~A. Cook.
\newblock Soundness and completeness of an axiom system for program
  verification.
\newblock {\em {SIAM} J. Comput.}, 7(1):70--90, 1978.

\bibitem{CytronR:effcssafcdg}
Ron Cytron, Jeanne Ferrante, BK~Rosen, Mark~N Wegman, and F.~K. Zadeck.
\newblock {Efficiently Computing Static Single Assignment Form and the Control
  Dependence Graph}.
\newblock {\em ACM Transactions on Programming Languages and Systems},
  13(4):451--490, October 1991.
%
%\bibitem{PintoJS:vercslip}
%Maria~Jo{\~a}o Frade and Jorge~Sousa Pinto.
%\newblock Verification {C}onditions for {S}ource-level {I}mperative {P}rograms.
%\newblock {\em Computer Science Review}, 5:252--277, 2011.

\bibitem{GordonMJC:forh}
Mike Gordon and Hélène Collavizza.
\newblock Forward with Hoare.
\newblock In {\em Reflections on the Work of C.A.R. Hoare}, History of Computing, pages
  101--121. Springer, 2010.

\bibitem{HoareCAR:axibcp}
C.~A.~R. Hoare.
\newblock An axiomatic basis for computer programming.
\newblock {\em Communications of the ACM}, 12:576--580, 1969.

\bibitem{PintoJS:sinapv}
Cl{\'{a}}udio~Belo Louren{\c{c}}o, Maria~Jo{\~{a}}o Frade, and
Jorge~Sousa  Pinto.
\newblock Formalizing Single-assignment Program Verification: an
Adaptation-complete Approach.
\newblock To appear in Proceedings of {\em ESOP 2016}.

\bibitem{ReynoldsJC:thepl}
John~C. Reynolds.
\newblock {\em Theories of Programming Languages}.
\newblock Cambridge University Press, Cambridge, England, 1998.

\end{thebibliography}

\end{document}